\documentclass[sigconf, 9pt]{acmart}
\usepackage{enumitem}
\usepackage{svg}
\usepackage{xspace}
\newcommand{\defn}[1]{{\textit{\textbf{\boldmath #1}}}\xspace}

\newcommand{\waste}{\texttt{waste}}

\newcommand{\simple}{\textsf{SIMPLE}\xspace}
\newcommand{\trash}{\textsf{RSUM}\xspace}

\newcommand{\geo}{\textsf{GEO}\xspace}

\newcommand{\tinyhash}{\textsf{TINYHASH}\xspace}
\newcommand{\flexhash}{\textsf{FLEXHASH}\xspace}

\usepackage{bbm}

\newcommand{\bigO}{O}
\DeclareMathOperator{\E}{\mathbb{E}}

\DeclareMathOperator{\polylog}{\text{polylog}}
\DeclareMathOperator{\poly}{\text{poly}}

\newcommand{\interior}[1]{ {\kern0pt#1}^{\mathrm{o}} }

\newcommand{\eps}{\varepsilon}
\renewcommand{\epsilon}{\varepsilon}
\newcommand{\setof}[2]{\left\{ #1\; \mid \;#2 \right\}}
\newcommand{\set}[1]{\left\{ #1\right\}}

\newcommand{\I}{\mathcal{I}}
\renewcommand{\H}{\mathsf{H}}
\newcommand{\R}{\mathbb{R}}

\newcommand{\Z}{\mathbb{Z}}

\newcommand{\N}{\mathbb{N}}

\usepackage{algorithm}
\usepackage[noend]{algpseudocode} % adding [noend] deletes the end while and stuff
\makeatletter
\newcounter{HALG@line}
\renewcommand{\theHALG@line}{\thealgorithm.\arabic{ALG@line}}
\makeatother

\newcommand{\floor}[1]{\lfloor #1 \rfloor}

\newcommand{\ceil}[1]{\lceil #1 \rceil}
\newcommand{\dceil}[1]{\lceil #1 \rceil}
\newcommand{\paren}[1]{\left( #1 \right)}

\usepackage[capitalise,nameinlink,noabbrev]{cleveref}
\crefname{equation}{}{} % cref{eq:blah} only does (1) instead of Equation (1)
\crefname{enumi}{Step}{} % cref{eq:blah} only does (1) instead of Item(1)

\newtheorem{fact}{Fact}

\makeatletter
\newtheorem*{rep@theorem}{\rep@title}
\newcommand{\newreptheorem}[2]{%
    \newenvironment{rep#1}[1]{%
        \def\rep@title{\cref{##1}}%
        \begin{rep@theorem}}%
            {\end{rep@theorem}}}
\makeatother
\newreptheorem{theorem}{Theorem}

\usepackage{thmtools}
\usepackage{thm-restate}

\ccsdesc[500]{Theory of computation~Design and analysis of algorithms}
\keywords{Memory Reallocation}

\copyrightyear{2024}
\acmYear{2024}
\setcopyright{licensedusgovmixed}\acmConference[SPAA '24]{Proceedings of the 36th ACM Symposium on Parallelism in Algorithms and Architectures}{June 17--21, 2024}{Nantes, France}
\acmBooktitle{Proceedings of the 36th ACM Symposium on Parallelism in Algorithms and Architectures (SPAA '24), June 17--21, 2024, Nantes, France}
\acmDOI{10.1145/3626183.3659965}
\acmISBN{979-8-4007-0416-1/24/06}

\acmConference[SPAA '24]{}{June, 17-21, 2024}{Nantes, France}

\author{Martin Farach-Colton}
\authornote{This work was supported in part by NSF grants CNS-2118620 and CCF-2106999.}
\affiliation{
 \institution{New York University}
 \city{New York}
 \state{NY}
  \country{USA}
}
\email{martin.farach-colton@nyu.edu}

\author{William Kuszmaul}
\affiliation{
\institution{Harvard University}
\city{Cambridge}
\state{MA}
\country{USA}
}
\email{william.kuszmaul@gmail.com}
\authornote{
William Kuszmaul is funded by the Rabin Postdoctoral Fellowship in Theoretical Computer Science at Harvard University. Large parts of this research were completed while William was a PhD student at MIT, where he was funded by a Fannie and John Hertz Fellowship and an NSF GRFP Fellowship. William Kuszmaul was also partially sponsored by the United States Air Force Research Laboratory and the United States Air Force Artificial Intelligence Accelerator and was accomplished under Cooperative Agreement Number FA8750-19-2-1000. The views and conclusions contained in this document are those of the authors and should not be interpreted as representing the official policies, either expressed or implied, of the United States Air Force or the U.S. Government. The U.S. Government is authorized to reproduce and distribute reprints for Government purposes notwithstanding any copyright notation herein.
}
\author{Nathan Sheffield}
\affiliation{
\institution{Massachusetts Institute of Technology}
\city{Cambridge}
\state{MA}
\country{USA}
}
\email{shefna@mit.edu}
\author{Alek Westover}
\affiliation{
\institution{Massachusetts Institute of Technology}
\city{Cambridge}
\state{MA}
\country{USA}
}
\email{alekw@mit.edu}

\title{A Nearly Quadratic Improvement for Memory Reallocation}
\begin{document}

\begin{abstract}
    In the Memory Reallocation Problem a set of items of various sizes must be dynamically assigned to non-overlapping contiguous chunks of memory.  It is guaranteed that the sum of the sizes of all items present at any time is at most a $(1-\eps)$-fraction of the total size of memory (i.e., the load-factor is at most $1-\epsilon$). The allocator receives insert and delete requests online, and can re-arrange existing items to handle the requests, but at a \emph{reallocation cost} defined to be the sum of the sizes of items moved divided by the size of the item being inserted/deleted. 

The folklore algorithm for Memory Reallocation achieves a cost of $O(\epsilon^{-1})$ per update.  In recent work at FOCS'23, Kuszmaul showed that, in the special case where each item is promised to be smaller than an $\epsilon^4$-fraction of memory, it is possible to achieve expected update cost $O(\log\epsilon^{-1})$. Kuszmaul conjectures, however, that for larger items the folklore algorithm is optimal. 

In this work we disprove Kuszmaul's conjecture, giving an allocator that achieves expected update cost $O(\epsilon^{-1/2} \operatorname*{polylog} \epsilon^{-1})$ on any input sequence. We also give the first non-trivial lower bound for the Memory Reallocation Problem: we demonstrate an input sequence on which any resizable allocator (even \emph{offline}) must incur amortized update cost at least $\Omega(\log\epsilon^{-1})$. 

Finally, we analyze the Memory Reallocation Problem on a stochastic sequence of inserts and deletes, with random sizes in $[\delta, 2 \delta]$ for some $\delta$. We show that, in this simplified setting, it is possible to achieve $O(\log\epsilon^{-1})$ expected update cost, even in the ``large item'' parameter regime ($\delta > \epsilon^4$).
\end{abstract}
\maketitle

\section{Introduction}

In the \defn{Memory Reallocation Problem} an \defn{allocator} must assign a dynamic set of items to non-overlapping contiguous chunks of memory. Given an set of items with sizes $x_1, x_2, \ldots, x_n$, and given a memory represented by the real interval $[0, 1]$,  a \defn{valid allocation} of these items to memory locations is a set of locations $y_1,\ldots, y_n\in [0,1]$ so that the intervals $(y_i, y_i+x_i)\subset [0,1]$ are all disjoint. As objects are inserted/deleted over time, the job of the allocator is rearrange items in memory so that, at any given moment, there is a valid allocation. The allocator is judged by two metrics: the maximum \defn{load factor} that it can support; and the \defn{reallocation overhead} that it induces. The allocator is said to support \defn{load factor} $1 - \eps$ if it can handle an arbitrary sequence of item insertions/deletions, where the only constraint is that the sum of the sizes of the items present, at any given moment, is never more than $1 - \eps$; and the allocator is said to achieve \defn{overhead} (or \defn{cost}) $c$ on a given insertion/deletion, if the sum of the sizes of the items that are rearranged is at most a $c$-factor larger than the size of the item that is inserted/deleted. We remark that all of the allocators in this work will be \defn{resizable}, meaning that if $L \le 1-\eps$ is the total size of items present at any time then, then all the items are placed in the interval $[0,L+\eps]$.

The Memory Reallocation Problem, and its variations, have been studied in a variety of different settings, ranging from history independent data structures \cite{naor2001anti, Ku23}, to storage allocation in databases \cite{bender2017cost}, to allocating time intervals to a dynamically changing set of parallel jobs \cite{bender2015cost, bender2013reallocation, lim2015dynamic}. The version considered here \cite{naor2001anti, Ku23, bender2015cost} is notable for its choice of cost function: if we model the \emph{time} needed to allocation/deallocate/move an object of size $s$ as $O(s)$, then an overhead of $O(c)$ implies that the total time spent moving objects around is at most an $O(c)$-factor larger than the time spent simply allocating/deallocating objects. The problem of minimizing movement overhead is especially important in systems with many parallel readers, since objects may need to be locked while they are being moved.

%%% WHAT DOES Otilde (1/log n) mean? and what is n?
\paragraph{Past Work} Most early work on memory allocation focused on the setting in which items \emph{cannot} be moved after being allocated (i.e., the 0-cost case) \cite{luby1996tight, robson1974bounds, robson1971estimate}. However, it is known that such allocators necessarily perform very poorly on their space usage -- they cannot, in general, achieve a load factor better than $\widetilde{O}(1 / \log n)$) \cite{luby1996tight, robson1974bounds, robson1971estimate}. The main goal in studying memory \emph{\textbf{re}}allocation \cite{Ku23, bender2017cost} is therefore to determine \emph{how much} item movement is necessary to achieve a load factor of $1 - \epsilon$.

The \defn{folklore algorithm} \cite{Ku23, bender2017cost} for the Memory Reallocation Problem is based on the
observation that whenever an item of size $k$ must be
inserted we can, by the pigeon-hole principle, find an
interval of size $O(k\eps^{-1})$ which has $k$ free space. Thus
it is possible to handle inserts at cost $O(\eps^{-1})$ and
handle deletes for free.

In recent work at FOCS'23 \cite{Ku23}, Kuszmaul shows
how to handle the case where all items have size smaller
than $\eps^4$ with expected update cost $\bigO(\log \eps^{-1})$.
However, Kuszmaul conjectures that, in general, the $O(\eps^{-1})$ folklore bound should be optimal. He proposes, in particular, that the special case in which objects have sizes in the range $(\epsilon, 2\epsilon)$ should require $\Omega(\eps^{-1})$ overhead per insertion/deletion. 

\paragraph{This Paper: Beating the Folklore Bound} In this work we disprove Kuszmaul's conjecture. In fact, we prove a stronger result: that it is possible to beat the folklore $O(\eps^{-1})$ bound without any constraints on object sizes.

We begin by considering the specialized setting in which items have sizes in the range $(\epsilon, 2 \epsilon)$---this, in particular, was the setting that Kuszmaul conjectured to be hard. We give in \cref{sec:poc} a relatively simple allocator that achieves $\bigO(\eps^{-2/3})$ amortized update cost in the case where all items have sizes in $(\eps, 2\eps)$. Although this allocator does not solve the full problem that we care about, it does introduce an important algorithmic idea that will be useful throughout the paper: the idea of having a special small set of items stored as a suffix of memory which are each ``responsible'' for a large number of items in the main portion of memory. Whenever an item from the main portion of memory is deleted, it gets ``replaced'' with an item that was responsible for it from the small suffix of memory. By using this notion of responsibility in the right way, we can imbue enough combinatorial structure into our allocation algorithm that it is able to beat the folklore $O(\epsilon^{-1})$ bound.

The construction of \cref{sec:poc} is a good start, but does not immediately generalize to handle arbitrary item sizes. In 
\cref{sec:geo} we give several new ideas to handle the case of items with sizes in $[\eps^5, 1]$. Then, we show how to combine this allocator with Kuszmaul's allocator from \cite{Ku23} to achieve:
\begin{reptheorem}{cor:onehalf}
There is a resizable allocator for arbitrary items with expected
update cost $\tilde{\bigO}(\eps^{-1/2}) = \bigO(\eps^{-1/2} \operatorname*{polylog} \epsilon^{-1}).$
\end{reptheorem}
At a high level, the algorithm in \cref{cor:onehalf} takes the basic idea from \cref{sec:poc} (a small suffix of items that take responsibility for items in the main array), and applies it in a nested structure. This nested ``responsibility'' structure is not simply a recursive application of the technique---rather, it is carefully constructed so that items of a given size can only appear some levels of the nest. This ends up being what enables us to beat the folklore bound with an arbitrary combination of item sizes.

%We remark that, although our algorithm for \cref{cor:onehalf} offers a randomized guarantee, the guarantee is nonetheless true even against an adaptive adversary. That is, even if the sequence of insertions/deletions is made by an adversary that can see the state of memory over time, the allocation will still achieve an expected amortized update cost of $\widetilde{\bigO}(\eps^{-1/2})$. This was not true of the $O(\log \epsilon^{-1})$ bound achieved by past work for the small-objects case \cite{Ku23}.

%\todo{Do we also get a non-amortized guarantee for oblivious sequences if we want it? I don't think so. if they want to make a particular operation cheap I think they can do it. ANSWER: no}

We conclude the paper with two additional results that are of independent interest. The first is a lower bound, showing that $O(1)$ update cost is not, in general, possible. And the second is an upper bound for a special case where the input sequence is generated by a simple stochastic process.

Until now, the only non-trivial lower bounds for the Memory Reallocation Problem have been for very restricted sets of allocation algorithms \cite{Ku23}. In \cref{sec:short-lowerbound}, we give a lower bound that applies to any (even offline) allocator. In fact, the update sequence which we use to establish the lower bound is remarkably simple, involving just two item sizes.
\begin{reptheorem}{thm:ztlower}
There exist sizes $s_1,s_2 \in \Theta(\eps^{1/2})$ and an update sequence $S$ consisting solely of items of sizes $s_1,s_2$ such that any resizable allocator (even one that knows $S$) must have amortized update cost at least $\Omega(\log\eps^{-1})$ on $S$.
\end{reptheorem}

Finally, in \cref{sec:random}, we consider a setting where item arrivals and departures follow a simple stochastic assumption. Define a \defn{$\delta$-random-item sequence} as one where memory is first filled with items of sizes chosen randomly from $[\delta, 2\delta]$, and then the allocator receives alternating deletes of random items and inserts of items with sizes chosen randomly from $[\delta, 2\delta]$. In this setting we are able to achieve $O(\log \eps^{-1})$ overhead:
\begin{reptheorem}{thm:summer}
For any $\delta = \poly(\eps)$, there is a resizable allocator that handles $\delta$-random-item sequences with worst-case expected update cost $\bigO(\log\eps^{-1})$.
\end{reptheorem}
We note that the algorithm for \cref{thm:summer} uses very different techniques from the other algorithms proposed in the paper. In fact, because of this, the algorithm in \cref{thm:summer} ends up being quite nontrivial to implement time-efficiently. We give an implementation that decides which items to move in worst-case expected time $\bigO(\eps^{-1/2})$ per update. The time bound is due to a technically interesting lemma about subset sums of random sets.

% In \cref{sec:small} we give an allocator for items of sizes in $\eps^{2+\gamma}, 2\eps^{2+\gamma}$ that achieves amortized update cost $\bigO(\eps^{-1})$. We show how to combine a variant of this strategy this with Kuszmaul's earlier construction to obtain an allocator for items with maximum size $\eps^2$ that achieves amortized update cost $\bigO(\log \eps^{-1})$. 
% The major remaining open question for the small items case, i.e., all items have size at-most $\eps^2$  is whether $\bigO(\log\eps^{-1})$ is tight, or if the $\bigO(1)$-cost allocators can be extended to apply to this entire regime.

\section{Preliminaries and Conventions}\label{sec:prelim}

We use $[n]$ to denote the set $\set{1,2,\ldots, n}$. 
For set $X$ and value $y$ we define $y+X = \setof{y+x}{x\in X}$ and $y\cdot X = \setof{yx}{x\in X}$.
We use $\log$ to denote $\log_2$.
We use $|I|$ to denote the size of an item $I$. The \defn{total size} of a set of items is defined to be the sum of their sizes. 
We will refer to memory as going from left to right, i.e., the start of memory is on the left and the end of memory is on the right.

In the \defn{Memory Reallocation Problem} with free-space parameter $\epsilon$, an \defn{allocator} maintains a set of items in memory, which is represented by the interval $[0,1]$.
Memory starts empty, and items are inserted and deleted over time by an oblivious adversary, where the only constraint on the update sequence is that the items present at any time must have total size at most $1-\eps$. The job of an allocator is to maintain a dynamic allocation of items to memory, that is, to assign each item to a disjoint interval whose size equals the item's size. If the allocator moves $L$ total size of items on an update of size $k$ we say the update is handled at \defn{cost} $L/k$.

We construct allocators that give an extra guarantee: If $L\in [0, 1-\eps]$ is
the total size of items present at any time, then a \defn{resizable allocator}
guarantees that all the items are placed in the interval $[0,L+\eps]\subseteq
[0,1]$.

Our analysis is asymptotic as a function of $\eps^{-1}$.
Thus, we may freely assume that $\eps^{-1}$ is at least a
sufficiently large constant.
We use the notation $\widetilde{O}$ to hide $\polylog(\eps^{-1})$ factors,
and the notation $\poly(n)$ to denote $n^{\Theta(1)}$.

\section{An Allocator for Large Items}\label{sec:poc} In this section we describe a simple allocator for a special
case of the Memory Reallocation Problem, disproving a conjecture of
Kuszmaul \cite{Ku23}.
    We remark that the folklore bound only gives performance
    $\bigO(\eps^{-1})$ in the regime of \cref{thm:poc}, i.e., gives no non-trivial bound.
\begin{theorem}\label{thm:poc}
There is a resizable allocator for items of with sizes in  $[\eps,2\eps)$ that
achieves amortized update cost $\bigO(\eps^{-2/3})$.
\end{theorem}

Theorem \ref{thm:poc} offers an \emph{amortized bound}, although, as we shall
see in \cref{sec:geo}, it is also possible to obtain a non-amortized expected
bound. We remark that there are two notions of amortized cost that one could
reasonably consider -- if $L_i$ denotes the total-size of items moved to handle
the $i$-th update and $k_i$ is the size of the $i$-th update, then either of
$\frac{1}{n} \sum_{i = 1}^n L_i / k_i$ or $\sum_{i = 1}^n L_i / \sum_{i = 1}^n
k_i$ would be a reasonable objective function. Fortunately, in this section,
because the $k_i$'s are all equal up to a factor of two, the two objective
functions are the same up to constant factors. In later sections where object
sizes differ by larger factors, we will go with the convention that guarantees
should be worst-case expected rather than amortized.

\begin{proof}
    We call our allocator \simple. We partition the sizes $[\eps,2\eps)$ into
    $\lceil\eps^{-1/3}\rceil$ \defn{size classes}, where the $i$-th size class
    consists of items with size in the range $$[\eps+(i-1)\eps^{4/3},
    \eps+i\eps^{4/3}).$$ Now we describe the operation of \simple; we also
    provide pseudocode for \simple in \cref{alg:poc}.
    
\begin{algorithm}
    \centering
    \caption{\simple Allocator}\label{alg:poc}
    \begin{algorithmic}[1]
        \State \simple maintains a suffix of the items called the \defn{covering set}.
        \If{it has been $\lfloor \eps^{-1/3} \rfloor$ updates since the last rebuild (or it is the first update)}
        \State Perform a rebuild as follows:
        \State Logically restore items to their original size (i.e.,  revert any logical inflation of sizes).
        \State For each $i\in [\lceil\eps^{-1/3}\rceil]$ let $x_i$ be the number of items of the $i$-th
        size class.
        \State Let $S$ be the union over $i\in[\lceil\eps^{-1/3}\rceil]$ of
        the smallest $\min(x_i, \lfloor\eps^{-1/3}\rfloor)$ items in the $i$-th
        size class.
        \State Arrange items to be contiguous and left-aligned, with items $S$ occurring after the other items.
        \State Update the covering set to be $S$.
        \EndIf
        \If{an item $I$ is inserted}
        \State Place $I$ immediately after the final
        item of the covering set and add $I$ to the covering set.
        \ElsIf{an item $I$ is deleted}
        \If{$I$ is not part of the covering set}
        \State Let $I'$ be an item from the covering set of
        the same size class as $I$ with $|I'|\le |I|$. 
        \State Place $I'$ where $I$ used to start.
        \State Logically inflate the size of $I'$ to $|I|$.
        \EndIf
        \State{Remove $I$ from memory.}
        \State{Compact the covering set, arranging its items to be contiguous and flush with the non-covering set.}
        \EndIf
    \end{algorithmic}
\end{algorithm}

    \paragraph{Rebuilds}
    Every $\lfloor \eps^{-1/3}\rfloor$ updates (starting from the first update)
    \simple performs a \defn{rebuild}. 
    Let $x_i$ be the number of items of size class $i$ at the
    time of this rebuild. In a rebuild operation \simple takes the $\min(x_i,
    \lfloor \eps^{-1/3}\rfloor)$ smallest items from size class $i$ for each
    $i\in [\lceil\eps^{-1/3}\rceil]$ and groups them into a \defn{covering set}.
    \simple arranges memory so that the items are contiguous, left-aligned
    (i.e., starting at $0$), and so that the covering set is a suffix of the
    present items.
    
    \paragraph{Handling inserts}
    When an item is \textbf{inserted} \simple adds the item to the covering set
    and places it directly after the final element currently in memory. 
    
    \paragraph{Handling deletes}
    Suppose an item $I$ of size class $i$ is \textbf{deleted}. If $I$ is not
    part of the covering set \simple finds an item $I'$ in the covering set
    which is also of size class $i$ but with $|I'|\le |I|$. \simple places $I'$
    at the location where $I$ used to start 
    and \defn{logically inflates} item $I'$ to be of size $|I|$. That is,
    \simple will consider item $I'$ to be of size $|I|$ until $I'$ is inflated
    even further or until the next rebuild. On each rebuild all items are
    reverted to their actual size. We say this \defn{swap} operation introduces
    \defn{waste} $|I|-|I'|\le \eps^{4/3}$ into memory. Finally,
    regardless of whether $I$ was in the covering set, \simple ends the delete
    by removing $I$ from memory and \defn{compacting} the covering set, i.e.,
    arranging the covering set items to be contiguous, and left-aligned against
    the end of the non-covering-set.

\begin{figure}
    \centering
    \includegraphics[width=.9\linewidth]{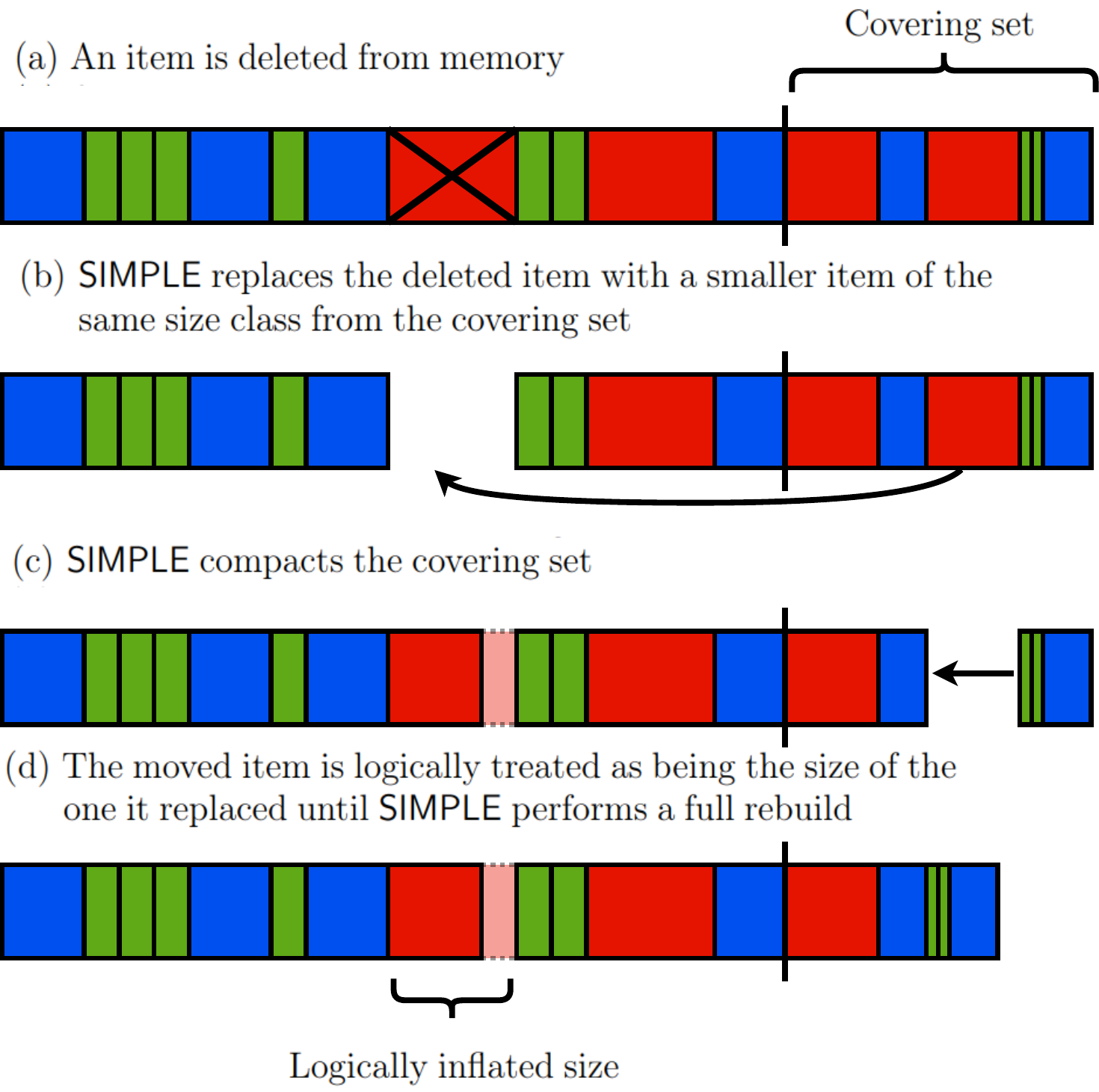}
    \caption{A depiction of \simple handling a delete of an item $I$ outside of the covering set by replacing $I$ with an item $I'$ from the covering set, inflating $I'$ to size $|I|$, and compacting the covering set.}
    \label{fig:poc-pic}
\end{figure}

\begin{lemma}\label{lem:simplecorrect}
\simple is correct and well-defined. 
\end{lemma}
\begin{proof}
    To verify correctness we must show that \simple places items within the
    allowed space. \simple essentially stores the items contiguously, except for
    the waste introduced on deletes. Each delete creates waste at most
    $\eps^{4/3}$: the maximum possible size difference between two items of the
    same size class. \simple performs a rebuild every $\lfloor \eps^{-1/3}
    \rfloor$ updates. Thus, the total waste in memory will never exceed
    $$\lfloor\eps^{-1/3}\rfloor \cdot \eps^{4/3} \le \eps.$$
    Thus, if the total size of items present is $L$, \simple stores all items in
    the memory region $[0, L+\eps]$.

    To verify that \simple is well-defined we must argue that on every delete of
    an item outside of the covering set \simple can find a suitable item in the
    covering set to swap with the deleted item; all other parts of \simple's
    instructions clearly succeed. Fix a size class $i$. We consider two
    (exhaustive) cases for how many items of size class $i$ were placed in the
    covering set on the previous rebuild, and argue that in either case whenever
    an item $I$ of size-class $i$ outside the covering set is deleted \simple
    can find an appropriate item $I'$ in the covering set to swap with $I$.

    \textbf{Case 1}: The $\lfloor \eps^{-1/3}\rfloor$ smallest items of size
    class $i$ were placed in the covering set on the previous rebuild; call this
    set of items $S_i$. Then, because \simple performs rebuilds every $\lfloor
    \eps^{-1/3}\rfloor$ updates and because \simple swaps at most one of the
    items from $S_i$ out of the covering set on each delete we have that on any
    delete before the next rebuild there is always an element of $S_i$ contained
    in the covering set. The items in $S_i$ were chosen to be the smallest items
    of size class $i$ at the time of the previous rebuild. Recall that inserted
    items are added to the covering set. Thus, we maintain the invariant that
    all items $I$ of size class $i$ outside of the covering set have (logical)
    size at least the size of any element in $S_i$. Thus, there is always an
    appropriate covering set item to swap with any deleted item of size class
    $i$ outside of the covering set.

    \textbf{Case 2}: If we are not in Case 1, then during the previous rebuild
    there were fewer than $\lfloor \eps^{-1/3} \rfloor$ total items of size
    class $i$, and \simple placed \emph{all of these items} in the covering set.
    This property, that all items of size class $i$ are contained in the
    covering set, is maintained until the next rebuild because inserted items
    are added to the covering  set. Thus, until the next rebuild there is
    \emph{never} a delete of an item of size class $i$ outside of the covering
    set: no such items exist. So the condition we desire to hold on such deletes
    is vacuously true.
    
\end{proof}

% \todo{font of figures is different than paper}
\begin{lemma}
\simple has amortized update cost $\bigO(\eps^{-2/3})$.
\end{lemma}
\begin{proof}
The covering set has size at most $2\eps \cdot \lceil \eps^{-1/3} \rceil \cdot
2\lfloor\eps^{-1/3}\rfloor \le \bigO(\eps^{1/3})$. This is because all items
have size at most $2\eps$, the number of size classes is $\lceil  \eps^{-1/3}
\rceil$, and the number of items of each size class in the covering set starts
as at most $\lfloor \eps^{-1/3} \rfloor$ and then increases by at most one per
update during the $\lfloor \eps^{-1/3} \rfloor$ updates between rebuilds, and
hence the number of items of each size class in the covering set never exceeds
$2\lfloor \eps^{-1/3} \rfloor$. We compact the covering set on each update and
so incur cost $\bigO(\eps^{1/3}/\eps)\le \bigO(\eps^{-2/3})$ per update.
Rebuilds incur cost at most $1/\eps$, and occur every $\lfloor
\eps^{-1/3}\rfloor$ steps. Thus, their amortized cost is at most
$\eps^{-1}/\lfloor \eps^{-1/3}\rfloor \le \bigO(\eps^{-2/3})$. Overall,
\simple's amortized update cost is $\bigO(\eps^{-2/3})$.
\end{proof}

\end{proof}

\section{An Allocator for Arbitrary Items} \label{sec:geo}
\cref{thm:poc} gives a surprising and simple demonstration that the folklore bound is not tight in the large items regime. In this section we will show how to outperform the folklore algorithm for arbitrary items, which is substantially more difficult than \cref{thm:poc}. In \cite{Ku23} Kuszmaul has already shown how to outperform the folklore algorithm in the regime where items are very small. 
In \cref{subsec:combine} we show that Kuszmaul's allocator can be combined with any resizable allocator fairly easily, to even get a resizable allocator.
Thus, the main difficulty we address in this section is extending
\cref{thm:poc}'s allocator \simple 
to work on items with sizes in the interval $[\eps^5, 1]$.
There are two major obstacles not present in \simple that arise when
handling items with sizes that can differ by factor of $\poly(\eps)$.

The first challenge is that \simple
compacts the entire covering set on every delete. 
The covering set needs to be large enough to contain a substantial quantity of
items of each size class. Large items, e.g., of size close to $\eps^{1/2}$ can
potentially afford to compact the covering set each time they are the subject of
an update. However, it would be catastrophic if updates of smaller items, e.g.,
items of size $\eps^3$ caused the entire covering set to be compacted each time.
In fact, the situation is even more troublesome: we hope to improve \simple's update cost of $\bigO(\eps^{-2/3})$ to $\widetilde{\bigO}(\eps^{-1/2})$.
Thus, even items of size $\Theta(\eps)$ cannot afford to compact the entire covering set on each update if the covering set is large. 
And, in order to make rebuilds infrequent it seems like we must make the covering set quite large. 
% We address this issue by recursively nesting collections like the ``covering sets'' of \cref{thm:poc}, while putting different amounts of each item in the nested covering sets based on their size. 
 
The second challenge is that \simple breaks items into size classes, which are
groups of items whose sizes differ by at most $\eps^{4/3}$. The small
multiplicative range of item sizes that we assume in \cref{thm:poc} ensures that
the number of size classes will be small. 
However, we cannot use the same style of size classes once the item sizes can
vary by a factor of $\eps^5$: there would be far too many size classes.
In order to support a larger range of item sizes, we modify our size classes to
be \defn{geometric}. That is, we define size classes of the form
$[\delta(1+\alpha)^{i-1}, \delta(1+\alpha)^{i}]$ instead of $[\delta + \alpha (i-1),
\delta + \alpha i]$ for some $\alpha = \poly(\eps)$. 
However, geometric size classes cause a major complication absent in the
fixed-stride size class approach of \simple. Namely, with geometric size classes
large items waste more space than small items per delete. Thus, a naive approach
of rebuilding whenever the wasted space exceeds $\eps$ would be susceptible to
the following vulnerability: a few deletes of large items could waste a lot of
space, but then the rebuild could be triggered by a small item. But the rebuild
is very expensive when triggered by a small item.

We now introduce a construction to address these issues.

\subsection{Handling Items with Sizes in $[\eps^5,1]$}
\begin{theorem}\label{thm:geo}
  There is a randomized resizable allocator for items of size at least
  $\eps^5$ that achieves worst-case expected update cost $\widetilde{\bigO}(\eps^{-1/2}).$
\end{theorem}
\begin{proof}
  We call our allocator \geo. 
  \geo labels an item as \defn{huge} if it has size at least $\eps^{1/2}/100$. 
  Whenever a huge item $I$ is inserted or deleted
  \geo rearranges all of memory so that all huge 
  items are compacted together at the start of memory. 
  The cost of each such operation is $\bigO(\eps^{-1/2})$. 
  Thus, we may assume without loss of generality that there are no huge items.
  Assume that $\eps^{-1}$ is a power of $4$. 
  This is without loss of generality up to decreasing $\eps$ by at most a factor of $4$.
  
  Let $\beta = 1+\eps^{1/2}$. \geo classifies the non-huge items into $C\le
  \bigO(\eps^{-1/2}\log \eps^{-1})$ size classes. Specifically, an item is
  classified as part of the $i$-th size class if it has size in the interval
  $[\eps^{5}\beta^{i-1}, \eps^{5}\beta^{i}).$ \geo builds a sequence of $\ell =
  4.5\log \eps^{-1}$ \defn{covering levels} \footnote{Note that $\ell\in \N$ by
  our assumption that $\eps^{-1}$ is a power of $4$.} -- nested suffixes of
  memory with geometrically decreasing sizes. In particular, if an item $I$ is
  in level $j$, we also consider $I$ to be in each level $j'<j$. For each $j\in
  [\ell]$ the \defn{mass limit} for each size class in level $j$ is defined to
  be $$m_j =  2^{\ell-j+1}\eps^5.$$ We will ensure the \defn{level size
  invariant}: for all $j\in [\ell], i\in [C]$ the total size of items of size class
  $i$ in level $j$ is at most $2m_j$. In particular, this will mean that the
  total size of level $j$ is at most $2C m_j$. 
  Note that $m_\ell = 2\eps^5$: the deepest level can fit only $O(1)$ of even
  the smallest items.  Also note that
  $$m_1 = 2^{\ell}\eps^{5} = 2^{4.5\log\eps^{-1}}\eps^{5} = \eps^{1/2},$$ 
  so level $1$ can fit at least $\Omega(1)$ of even the largest items. 
  For convenience we will also define the $0$-th level to
  mean all of memory with $m_0=1$. 
  
  For each $i\in [C]$, let $s_i$ denote the total number of items of size class
  $i$; this number will change as items are inserted and deleted. Let $b_i =
  \eps^{5}\beta^{i}$: all items of size class $i$ have size smaller than
  $b_i$. For each $i\in [C],j\in [\ell]$ 
  the number of items of size class $i$ in level $j$ will always be at most twice the quantity
  $$c_{i,j}=\floor{m_j/b_i}.$$ For
  convenience we also define $c_{i,0} = \infty$ for each $i\in [C]$. 

  % this is unnecessary: assume memory starts empty
  % \geo starts by placing the smallest $\min(s_i,c_{i,j})$ items of size class $i$ into level
  % $j$ for each $i\in [C]$ and each $j\in [\ell]\sqcup\set{0}$; recall that if an
  % item $I$ is in level $j$ we also consider the item to be in each level $j'<j$
  % so this is well-defined.

We now describe \geo.
% ; we also provide pseudocode for \geo in \cref{sec:appendixfigures}.
 
\paragraph{Level rebuilds}
For each $i\in [C]$, define $j_i^*$ to be the largest level $j\in [\ell]$ such
that $c_{i, j} \geq 1$; $j_i^*$ is the deepest level that could feasibly contain
an item of size class $i$. In fact we will have $c_{i,j_i^*} = 1$, because the
mass limit for any levels $j,j+1$ differ by a factor of $2$, and because the
mass limit in level $\ell$ is such that level $\ell$ fits at most $1$ of any
size class. For each $i\in [C], j \in [j_i^*]$ \geo keeps \defn{insert/delete
level rebuild thresholds} $r_{i,j}, r_{i,j}' \in [\dceil{c_{i,j}/4},
\dceil{c_{i,j}/3}]\cap \N$. \geo initializes the level rebuild thresholds
uniformly randomly from this range. 

Updates will sometimes cause \defn{level rebuilds}. 
To simplify the description of our allocator it is also useful to have a concept
of a \defn{free rebuild} (a type of rebuild). A free rebuild is a ``sentinel value'': it is only a
logical operation and has zero cost. At the very start \geo performs a free
rebuild of each level $j$ by each size class $i$. We describe a level rebuild
caused by an insert; level rebuilds caused by deletes are completely symmetric.
 Suppose an item of size class
$i_0$ is inserted. For each $j\in [\ell]$, let $t_j$ denote the number of
inserts since the previous time that level $j$ has been rebuilt (including free
rebuilds) by a size class $i_0$ item. Let $j_0$ be the smallest $j\in
[j_{i_0}^*]$ such that $t_j \ge r_{i_0,j}$ (in fact, we will have
$t_{j_0}=r_{i_0,j_0}$).

\geo then \defn{rebuilds} level $j_0$. 
For each $j \in [\ell], i\in [C]$ define $\mathcal{I}_j^{(i)}$ to be the $\min(s_{i},
c_{i,j})$ smallest items of size class $i$, and define
$$\mathcal{I}_j = \bigcup_{i\in [C]} \mathcal{I}_j^{(i)}.$$
Define $\overline{\mathcal{I}_j}$ to be all items except for items $\mathcal{I}_j$.
To rebuild, \geo rearranges level $j_0 - 1$ to ensure that for all $j \ge
j_0$ the items $\mathcal{I}_{j}$ appear to the right of items
$\overline{\mathcal{I}_j}$. This arrangement is
well-defined since for each $j$ we have $\mathcal{I}_{j + 1} \subseteq \mathcal{I}_{j}$. We 
justify in \cref{lem:geocorrect} why \geo can always find any such $\mathcal{I}_{j}$ as
a subset of level $j_0-1$, and so achieve this arrangement by rearranging only
level $j_0 - 1$. \geo labels the items $\mathcal{I}_{j}$ as level $j$ for all $j \ge j_0$.

Let $J$ be the set of all levels $j\in [j_i^*]$ such that $t_j \ge r_{i_0,j}$. To finish
the rebuild of level $j_0$ \geo resamples $r_{i_0,j}$ randomly from
$[\ceil{c_{i,j}/4}, \ceil{c_{i,j}/3}]\cap \N$ for each $j\in J$. \geo
considers this a free rebuild for levels $j\in J\setminus\set{j_0}$ by the size
class $i_0$ item.

\begin{algorithm}
\caption{Rebuild on an insert of item $I$}\label{alg:rebuilds-geo}
 \begin{algorithmic}[1]
\State Let $i_0$ denote the size class of item $I$.
\State  For each $j\in [\ell]$, let $t_j$ denote the number of inserts since the previous time that level $j$ has been rebuilt (including free rebuilds) by a size class $i_0$ item. 
\State Let $j_0$ be the smallest $j\in [j_{i_0}^*]$ such that $t_j \ge r_{i_0,j}$.
\State For each $j\in [\ell], i\in [C]$ define $\mathcal{I}_j^{(i)}$ to be the $\min(s_i, c_{i,j})$ smallest items of size class $i$.
\State  For $j\in [\ell]$ define $\mathcal{I}_j = \bigcup_{i\in [C]} I_j^{(i)}$ for all $j\in [\ell]$.
\State \textbf{Assert:} items $\mathcal{I}_j$ are present in level $j-1$
\For{$j\gets j_0, j_0+1, \ldots, \ell$}
    \State Arrange level $j-1$ so that items $\mathcal{I}_j$ are on the right, and other items are on the left.
    \State Label items $\mathcal{I}_{j}$ as level $j$.
\EndFor
\State Let $J$ be the set of all levels $j\in [j_i^*]$ such that $t_j \ge r_{i_0,j}$. 
\State Resample $r_{i_0,j}$ randomly from
$[\ceil{c_{i,j}/4}, \ceil{c_{i,j}/3}]\cap \N$ for each $j\in J$. 
\State \geo considers this a free rebuild for levels $j\in J\setminus\set{j_0}$ by the size
class $i_0$ item.
\end{algorithmic}
\end{algorithm}

\paragraph{Handling Inserts}
\geo handles inserts as follows: Place inserted items directly after the current
final item in memory. When an item of size class $i$ is inserted we add it to
level $\ell$. As discussed earlier, inserts trigger level rebuilds when level
rebuild thresholds are reached.

\begin{algorithm}
 \caption{Inserts}\label{alg:inserts-geo}
 \begin{algorithmic}[1]
    \State Place $I$ immediately after the final
    item of level $\ell$.
    \State Perform necessary level rebuilds.
 \end{algorithmic}
 \end{algorithm}

\paragraph{Handling deletes}
Suppose an item $I$ of size class $i$ is deleted. If item $I$ is not in level
$j_i^*$ \geo finds the item $I'$ of size class $i$ in level $j_i^*$ which will have
$|I'|\le |I|$ and \defn{swaps} $I,I'$; in \cref{lem:geocorrect} we argue that
there is some such item $I'$. To swap items $I$ and $I'$ \geo places item $I'$ where
item $I$ used to be. Next,
\geo\xspace \defn{inflates} the size of item $I'$ to $|I|$. That is, \geo will
logically consider item $I'$ to have size $|I|$ until the next \defn{waste
recovery} step at some later time (or until $I'$ is further inflated). We
describe the waste recovery procedure after finishing the description of how
\geo handles deletes.

After swapping item $I$ (if necessary) and removing $I$ from memory \geo\xspace
\defn{compacts} level $j^*_i$, i.e., arranges the items of level $j^*_i$ to be
contiguous and left-aligned with the final element that is not part of level
$j_i^*$ (or left-aligned with $0$ if all elements are part of the level). 
As discussed earlier a delete triggers level rebuilds when level rebuild
thresholds are reached.

 \begin{algorithm}
 \caption{Delete item $I$}\label{alg:delete-geo}
 \begin{algorithmic}[1]
 \Require $\waste$ so far and waste recovery threshold $T$.
  \State Remove item $I$ from memory.
  \State Let $i$ be the size class of item $I$.
  \State Let $j_i^*$ be the largest $j$ such that $c_{i,j} \geq 1$, i.e., so that $I$ fits in level $j$.
  \State \textbf{Assert:} the smallest item of size class $i$ is guaranteed to be in level $j_i^*$. \label{assert:geo-covering}
  \If{item $I$ is not in level $j_i^*$}
    \State Find the size class $i$ item $I'$ in level $j_i^*$.
    \State Place item $I'$ where item $I$ used to be.
    \State Logically inflate the size of
    $I'$ to be $|I|$.\label{line:geo-resizable-inflate} 
  \EndIf
  \State Let $b_i$ be the maximum size of an item of size class $i$.
  \State $\waste \gets \waste + \eps^{1/2}b_i$.
  \State Compact level $j^*_i$.
  \State Perform necessary level rebuilds.
  \If{$\waste > T$}
  \State Perform a waste recovery step.
  \EndIf
 \end{algorithmic}
 \end{algorithm}

\begin{figure}
    \centering
    \includegraphics[width=1\linewidth]{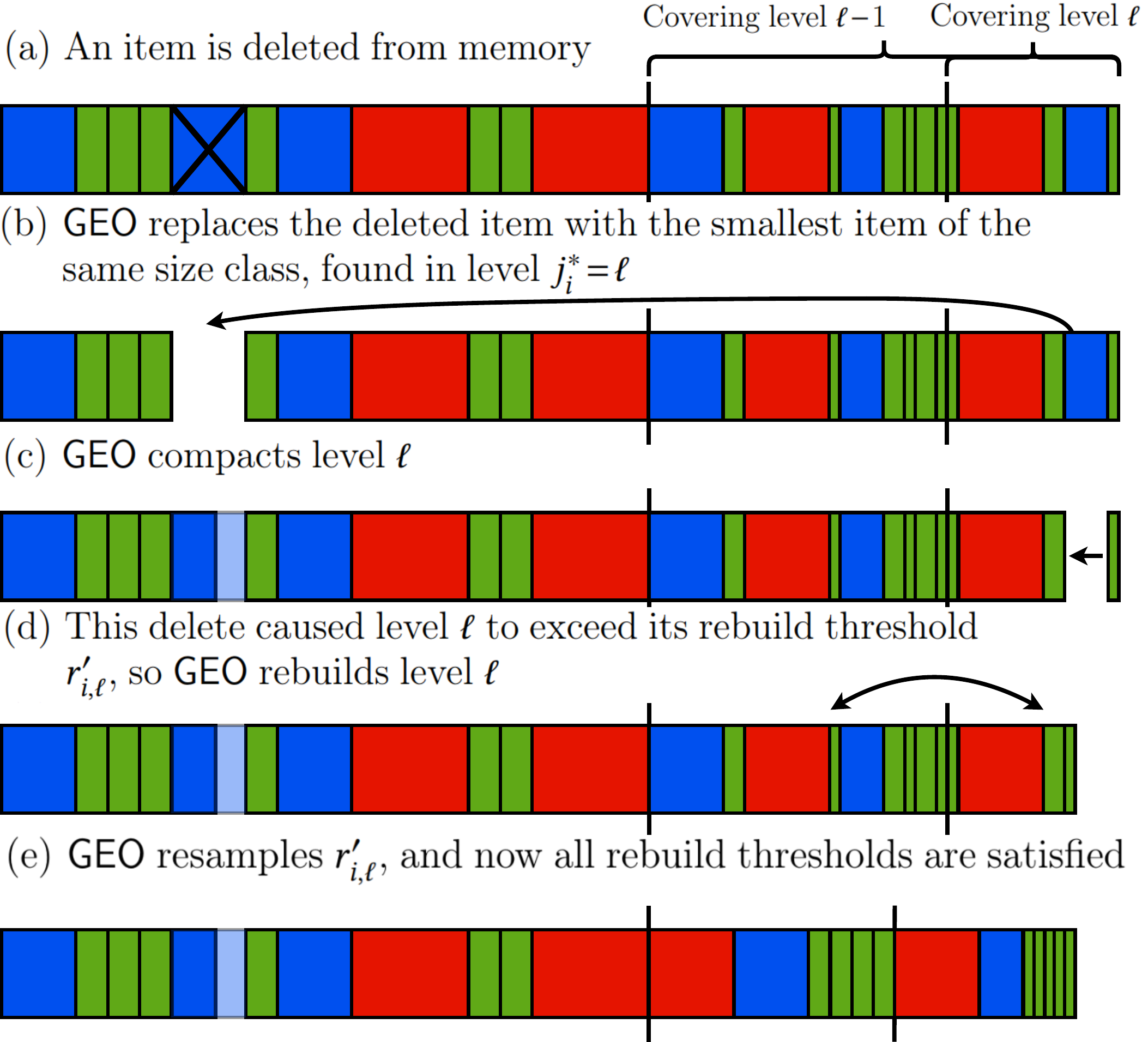}
    \caption{\geo handling a delete.}
    \label{fig:geo-pic}
\end{figure}

\paragraph{Implementing waste recovery}
When handling deletes \geo performs swaps which cause \defn{waste}.
Suppose \geo swaps items $I,I'$ both of size class $i$, and let $b_i$ be the
maximum size of an item in size class $i$.
Define $w_i = \eps^{1/2}b_i$. Then,
$|I|-|I'| \le b_i - b_i/\beta \le w_i$. We say that the swap causes waste $w_i$. 
\geo's waste recovery steps will ensure that the total waste in memory never
exceeds $\eps$. This will ensure that the total size of gaps introduced by swaps
never exceeds $\eps.$ 

We consider \geo to have performed a free waste recovery step at the beginning
(this is a logical operation incurring zero cost, useful as a sentinel value).
At every waste recovery step (and at the beginning) \geo
samples threshold $T\gets (\eps/2,\eps)$ uniformly to determine how much
waste to allow before triggering the next waste recovery step. More precisely,
(excluding the free waste recovery step at the beginning)
if the waste recovery threshold was $T$ and the most recent delete would cause
the waste introduced since the previous waste recovery step to be $W\ge T$ then
\geo performs a waste recovery step. We consider the waste at the start of this
waste recovery step to be $W-T$: that is, waste from the final delete which caused
the waste recovery step \defn{overflows} to count towards the next waste
recovery step. To perform the waste recovery step \geo logically reverts all
items to their original sizes, arranges the items to be contiguous
and left-aligned with $0$, and then rebuilds level $1$.

\begin{algorithm}
    \caption{Waste Recovery Step}\label{alg:waste-recovery-step-pseudo}
    \begin{algorithmic}[1]
      \State Revert all logical changes to item sizes.
      \State Compact all items to be contiguous and left-aligned.
      \State Rebuild level $1$.
      \State Let $W$ be the waste since the previous waste recovery step.
      \State $\waste\gets W-T$.
      \State Resample waste recovery threshold $T\gets (\eps/2, \eps)$.
    \end{algorithmic}
\end{algorithm}

\geo is depicted in \cref{fig:geo-pic}.
Now we analyze \geo.

\begin{lemma}\label{lem:geocorrect}
    \geo is well-defined and correct (i.e., allocates items within the allowed space).
\end{lemma}
\begin{proof}
First we show that the level size invariant is maintained. 
This follows from the following stronger property: for all $i\in [C], j\in
[\ell]$ there are at most $2c_{i,j}$ items of size class $i$ in level $j$.
First note that this is sufficient to prove the level size invariant because 
$2c_{i,j}$ items of size class $i$ take up at most $2m_j$ space. 
Now we argue that the rebuild procedure maintains this stronger property.
For all $i\in [C],j>j_i^*$, whenever an item of
size class $i$ is inserted some level $j\in [j^*_i]$ is rebuilt, and so no
items of size class $i$ can remain in level $j$ because $c_{i,j}=0$. For $i\in [C],j\in [j^*_i]$, 
the insert level rebuild threshold $r_{i,j}$ satisfies $r_{i,j}\le c_{i,j}$.
That is, level $j$ will be rebuilt before there are more than $c_{i,j}$ inserts of size class $i$ items, and thus level $j$ can never have more than $2c_{i,j}$ size class $i$ items. 

To show that \geo is correct, we need to verify that after every update
for every $j_0\in [\ell], j\ge j_0$, items $\mathcal{I}_j$ are contained in level $j_0-1$.
This is necessary for \geo's rebuild operation to be well-defined. 
Because for each $j$ we have $\mathcal{I}_{j+1} \subseteq \mathcal{I}_{j}$ it suffices to show that 
for each $j\in [\ell]$ the items $\mathcal{I}_j$ are contained in level $j-1$.
Recalling the definition of $\mathcal{I}_{j}$ our goal is to show that for all
$i\in [C], j\in [\ell]$ the $\min(s_i, c_{i,j})$ smallest items of size class
$i$ are contained in level $j-1$. Fix some size class $i$. 
First, observe that for $j>j_i^*$ we have $c_{i,j} = 0$, so the claim is
vacuously true. We prove the claim for $j\in [j_i^*]$ by induction on $j$.
The claim is clearly true for $j=1$: level $0$ is all of memory, so in
particular contains the $s_i$ smallest items of size class $i$. Assume the claim for $j\in
[j_i^*-1]$, we prove the claim for $j+1$.

Because the claim is true for $j$ we have that whenever level $j$ has just been
rebuilt it will contain the $\min(s_i, c_{i,j})$ smallest elements of size class
$i$, because these items were present in level $j-1$. 
We consider two cases.

\textbf{Case 1}: $c_{i,j}\le 3$. Then $r_{i,j}'=1$, i.e., level $j$ will be rebuilt every
time a size class $i$ item is updated. By our inductive hypothesis rebuilding
level $j$ results in the $\min(s_i,c_{i,j})\ge\min(s_i, c_{i,j+1})$ smallest
size class $i$ items being in level $j$, so the claim holds here.

\textbf{Case 2}: $c_{i,j}>3$. Then 
$$c_{i,j} - \ceil{c_{i,j}/3} \ge \ceil{c_{i,j}/2} \geq c_{i, j+1}.$$
Thus, if the smallest $c_{i,j}$ items of size class $i$ were placed in level $j$
on the previous level $j$ rebuild the smallest $c_{i,j+1}$ items of size class
$i$ will still be in level $j$ at all times until the next rebuild.
On the other hand, if the smallest $s_i$ items of size class $i$ were placed in
level $j$ on the previous level $j$ rebuild then no items of size class $i$ can
exit level $j$ until the next level $j$ rebuild: there are no size class $i$
items outside of level $j$ to trigger a swap.
Inserts are added to level $j$ so they do not break the invariant.
This proves the claim for $j+1$, so by induction the claim is true for all $j$.

In order for deletions to be well-defined, we must also show that after every
update for every size class $i$ with $s_i > 0$, the smallest element of size
class $i$ is in level $j_i^*$. This holds because we always have $r_{i,j_i^*}' =
1$, so every time level $j_i^*$ loses the smallest item of size class $i$ it will be
rebuilt, and when it is rebuilt it must have the smallest item because level
$j_i^*-1$ always contains items $\mathcal{I}_{j_i^*}$. All inserted items are
inserted to level $j_i^*$, so again insertions cannot break the invariant.

%This implies, in particular, that as long as there are any items of size class $i$ in memory, then the smallest (logically) item of size class $i$ is in level $j_i^*$, which ensures that the delete procedure is well-defined.

%To show that \geo is correct, we must verify the following assertion 
%made when describing the delete procedure:
%whenever an item $I$ of size class $i$ is deleted 
%there exists an item $I'$ in level $j_i^*$ with $|I'|\le |I|$.

%We argue inductively that the following invariant holds: if there are any items
%of size class $i$ in memory then the smallest (logically) item of size class $i$
%is in level $j_i^*$. The claim starts true, as memory starts empty. It is clear
%from the definition of a rebuild that after level $j_i^*$ is rebuilt either (a)
%there are no items of size class $i$ in memory, or (b) the smallest item of
%size class $i$ is contained in level $j$. Whenever an item of size class $i$ is
%deleted it will trigger a rebuild of level $j_i^*$ because the level rebuild
%threshold $r_{i,j_i^*}$ satisfies $r_{i,j_i^*}\le c_{i,j_i^*}=1$. Thus, deletes
%preserve the invariant. Whenever an item of size class $i$ is inserted it is
%placed in level $\ell$ and the insert will trigger a sequence of rebuilds
%eventually resulting in, at least, level $j_i^*$ being rebuilt. Thus, inserts
%also preserve the invariant.
    
Now, we argue that \geo always places items within the memory bounds. If we
consider items at their inflated (i.e., logical) sizes then the items are contiguous.
Recall that the total size of gaps introduced into the array by inflation is
bounded by the waste recovery threshold $T < \eps$. Hence, if there is $L$
total size of items present at some time \geo allocates all items in the memory
region $[0,L+\eps]$. That is, \geo is resizable.

For completeness we check the fact claimed when defining the size classes, that
$C\le \widetilde{\bigO}(\eps^{-1/2})$. Indeed, 
$$C \le \bigO(\log_\beta \eps^{-4.5}) \le \bigO\paren{\frac{\log\eps^{-1}}{\log
(1+\eps^{1/2})}}\le \bigO(\eps^{-1/2}\log\eps^{-1}).$$ 
\end{proof}

Before analyzing \geo's expected update cost we need two simple lemmas.
The proofs are deferred to \cref{sec:appendixobvious}.
\begin{restatable}{lemma}{continuousCats} \label{lem:hittingcontinuouscats}
    Fix $a,b,W\in \R$ with $0\le a<b$, and $W>0$.
    Let $x_1,x_2,\ldots$ be uniformly and independently sampled from $(W/2, W)$.
    The probability that there exists $j$ with $\sum_{i\le j} x_i \in [a,b]$ is
    at most $4(b-a)/W$.
\end{restatable}
\begin{restatable}{lemma}{discreteCats} \label{lem:hittingcats}
    Fix integers $y,N\in \N$.
    Let $x_1,x_2,\ldots$ be uniformly and independently sampled from $[\ceil{N/4},
    \ceil{N/3}]\cap \N$. The probability that there exists $j$ with $\sum_{i\le j} x_i
    = y$ is at most $100/N$.
\end{restatable}
    
  Now we analyze the worst-case expected cost of an update. For the remainder of
  the proof we fix an arbitrary update index $u\in \N$; our goal is to show that
  the expected cost on update $u$ is small. We break the cost of this update
  into $\Gamma_W + \Gamma_S + \Gamma_R$, where $\Gamma_W$ is the cost of waste
  recovery, $\Gamma_S$ is the cost of swapping elements and compacting
  to handle deletes, and $\Gamma_R$ is the cost of rebuilding levels. 
  We will show $\E[\Gamma_W+\Gamma_S+\Gamma_R]\le \widetilde{\bigO}(\eps^{-1/2})$.
  \begin{lemma}
    The expected cost due to waste recovery on update $u$
    satisfies $\E[\Gamma_W] \le\widetilde{\bigO}(\eps^{-1/2}).$
  \end{lemma}
  \begin{proof}
  % \geo only employs randomization in a single place: choosing the threshold $T$ for the allowable waste before triggering a waste recovery step. 
  % Thus, it suffices to consider an adversary that only is adaptive on waste recovery steps. In other words, the adversary gains no power by making decisions between waste recovery steps.
  % Thus, we can think of an adaptive adversary as fixing a sequence of updates 
  % that it is planning to perform  after every waste recovery step and only switching its strategy on the next waste recovery step.
  % We call the time between waste recovery steps a \defn{phase}.
  % Fix some phase. 
  % Let $x_1, x_2, \ldots,$ be the sequence of sizes of items that will be deleted. 
  % Note that the adversary may insert items as well, and may also delete items that it inserted earlier in the phase.
  % Now we bound the expected amortized cost of \geo's next waste recovery step based on $x_1,x_2,\ldots$.
  
  If update $u$ is an insert then \geo never performs a waste recovery step on update $u$. 
  Thus, for the purpose of analyzing the cost of waste recovery it suffices to
  consider the case that update $u$ is a delete.
  Let update $u$ be the $u'$-th
  delete, and let the corresponding deleted item be of size class $i$.
  Let $x_1, x_2, \ldots,$ be the sequence of sizes of items that will be deleted. 
  For each $k$, let $w_k$ be the space wasted by delete $k$, i.e., the maximum
  size difference between items in the size class of the $k$-th deleted item; we have $w_k \le O(\eps^{1/2}x_k)$.
  \geo repeatedly samples waste recovery thresholds $T_1, T_2,\ldots$ independently from $(\eps/2,\eps)$. 
  A waste recovery step occurs on update $u$ if there exists $M\in \N$ such that
  update $u$ causes the total waste to cross the $M$-th waste recovery threshold, i.e., 
  so that 
  $$\sum_{t=1}^{M}T_t \in \left[\sum_{k=1}^{u'-1} w_k, \sum_{k=1}^{u'}w_k\right].$$
  Here we have used the fact that waste \emph{overflows} between waste recovery steps. 
    By \cref{lem:hittingcontinuouscats} the probability that such an $M$ exists
    is at most $4w_{u'}/\eps$. 
    If $u$ must perform waste recovery the cost is at most $1/x_{u'}$.
    Thus, the expected cost of waste recovery on delete
    $u'$ is at most 
    $$\frac{4w_{u'}}{\eps} \frac{1}{x_{u'}} \le O(\eps^{-1/2}).$$

\end{proof}

\begin{lemma}
The cost due to swapping and compacting on update $u$ satisfies $\Gamma_S \le \widetilde{\bigO}(\eps^{-1/2})$.
\end{lemma}
\begin{proof}
When an item $I$ of size class $i$ is deleted \geo potentially 
moves an item $I'$ also of size class $i$ to replace item $I$. This costs $O(1)$.
After removing item $I$ from memory \geo must compact level $j_i^*$.
The cost of this compaction is bounded by the maximum possible size of level $j_i^*$ divided by $|I|$.
The size of level $j_i^*$ is at most $2Cm_{j_i^*}$ by the level size invariant.
We claim $|I|\ge m_{i,j_i^*}/4$. If $j_{i}^* < \ell$ but $|I|\le m_{i,j_i^*}/4$ then $I$'s size class can fit on a deeper level than $j_i^*$, contradicting the definition of $j_i^*$. If $j_i^* = \ell$ then the inequality is true because $m_{i, \ell}/4$ is smaller than the minimum item size.
Thus, the cost of compacting level $j_i^*$ is at most 
$$\frac{2C m_{j_i^*}}{|I|} \le \bigO(C)\le \widetilde{\bigO}(\eps^{-1/2}).$$
Note that there is zero cost here on an insert.
\end{proof}

\begin{lemma}\label{lem:rebuild-level-cost-geo}
The expected cost due to rebuilding levels on update $u$ satisfies $\E[\Gamma_R]\le \widetilde{\bigO}(\eps^{-1/2})$.
\end{lemma}
  \begin{proof}
  There are only $\ell = \Theta(\log\eps^{-1})$ levels. Thus, it suffices to fix
  a level $j\in [\ell]$ and show that the expected cost due to rebuilding
  level $j$ on update $u$ is at most $O(C)\le \widetilde{\bigO}(\eps^{-1/2})$.
  Fix $j\in [\ell]$ and let update $u$ be an item $I$ of size class $i\in [C]$.
  First, note that if $j>j_i^*$ level $j$ is never rebuilt by an item of size class $i$.
  So, we may assume $j\in [j_i^*]$.
% We consider two cases.
% \paragraph{Case 1: $j > j_i^*$}
% In this case $I$ is too large to fit in level $j$. In
% particular we have $|I| \ge m_j/4$. 
% Using the level size invariant to bound the size of level $j$, we have that the
% cost of rebuilding level $j$ is bounded by
% $$\frac{2Cm_{j}}{|I|}\le \bigO(C).$$
% \paragraph{Case 2: $j \le j_i^*$}
We claim the probability that update $u$ triggers a rebuild of level $j$ is
at most $100/c_{i,j}$.
Suppose that update $u$ is an insert; the case of deletes is symmetric.
Let $u$ be the $u'$-th insert of a size class $i$ item.
Let the sequence of insert rebuild thresholds $r_{i,j}$ for level $j$ on items of
size class $i$ chosen by \geo be $x_1,x_2,\ldots$. Recall that these are sampled
from $[\dceil{c_{i,j}/4}, \dceil{c_{i,j}/3}]\cap \N$.
Then, the probability of update $u$ triggering a rebuild of level $j$ is
precisely the chance that there is some $k^*$ such that $\sum_{k\le k^*} x_k 
 = u'$.
This is exactly the situation described in \cref{lem:hittingcats}. Thus, the
probability that $u$ triggers a rebuild of level $j$ is at most $100/c_{i,j}$ in
this case. 

If update $u$ triggers a rebuild of level $j$ the cost is at most $2Cm_j/|I|$
(and may even be $0$ in the case that it was a free rebuild, i.e., covered
by a larger level's rebuild). Thus, the expected cost of rebuilding level $j$ on
update $u$ is at most
\begin{equation}\label{eq:niceamorebuild}
\frac{200Cm_j}{c_{i,j}|I|}.
\end{equation}
Recall the definition of $c_{i,j}$: 
if $b_i$ denotes the maximum possible size in size class $i$ then $c_{i,j} =
\floor{m_j/b_i}$. Thus, because $c_{i,j}\ge 1$ we have
$$c_{i,j}\cdot |I|\ge c_{i,j}b_i/\beta = \floor{m_j/b_i}b_i/\beta \ge m_j/(2\beta)\ge m_j/4.$$
This shows that \cref{eq:niceamorebuild} is bounded by $O(C)$.
\end{proof}

Thus, the expected cost of update $u$ is at most
$$\E[\Gamma_S+\Gamma_W+\Gamma_R]\le \widetilde{\bigO}(\eps^{-1/2}).$$
  
\end{proof}

% \begin{remark}
% In-fact, the construction of \cref{thm:geo} can be de-randomized 
% with only $\polylog(\eps^{-1})$-times worse amortized update cost. 
% In particular, rather than having a randomized waste recovery step, 
% we could instead partition the $\eps$ free space amongst $\Theta(\log \eps^{-1})$ ``sub-pools'' 
% of items, which are groups of size classes consisting of items with relatively similar sizes.
% Then, a waste recovery step is triggered when the amount of space wasted by any sub-pool exceeds $\Omega(\eps/\log\eps^{-1})$.
% Doing so gives a deterministic allocator with
% amortized update cost $\widetilde{\bigO}(\eps^{-1/2})$.
% \end{remark}

\subsection{Combining \geo with Kuszmaul's Allocator}\label{subsec:combine}
Throughout the subsection we say that an item is \defn{large} if it has size larger
than $\eps^4$, and \defn{tiny} otherwise. In \cref{thm:geo} we described the
\geo allocator which can handle large items. In \cite{Ku23} Kuszmaul constructed
an allocator based on min-hashing, which we call \tinyhash, that can handle tiny
items with worst-case expected update cost $\bigO(\log\eps^{-1})$. Kuszmaul's
\tinyhash is even a resizable allocator, like \geo.
Combining \geo and \tinyhash immediately yields:

\begin{corollary}\label{cor:onehalf}
  There is an allocator for arbitrary items with worst-case expected update cost
  $\widetilde{\bigO}(\eps^{-1/2}).$
\end{corollary}
\begin{proof}
Instantiate \geo with $\eps/3$ free space starting at the beginning of memory
and instantiate \tinyhash with $\eps/3$ free space, 
but starting at the end of memory and growing backwards. 
When we get an update of a tiny item we send the update to \tinyhash, 
and when we get an update of a large item we send the update to \geo.
The correctness of this approach follows from the fact that \tinyhash and \geo are resizable. 
In particular, if at some time there is $L_1$ total size of tiny items 
present and $L_2$ total size of large items present then 
\geo only places items in the memory region $[0, L_1+\eps/3]$, 
and \tinyhash only places items in the memory region $[1-L_2-\eps/3, 1]$.
Because $L_1+L_2\le 1-\eps$ these intervals are disjoint.

This allocator inherits the $\max$ of the worst-case expected update costs in
$\geo,\tinyhash$ as its expected update cost.
\end{proof}

In fact, by exploiting the modular structure of \tinyhash, rather than simply
using \tinyhash as a black box we can strengthen \cref{cor:onehalf} to obtain
the same (asymptotically) update cost, but with a resizable allocator. Now, the
layout of memory will be space $[0,L_1+\eps/2]$ allocated to \geo, where $L_1$
is the total size of large items and then space $[L_1+\eps/2, L_1+L_2+\eps]$
allocated to \tinyhash where $L_2$ is the total size of tiny items present. As
before, \geo handles large items and \tinyhash handles tiny items. The
difference now is that \tinyhash doesn't have a fixed start location: as the
region of memory managed by \geo changes size we have to ensure that the region
of memory managed by \tinyhash starts right after \geo's memory region ends. 
That is, in addition to the usual \defn{internal updates}, we have to modify
\tinyhash to support \defn{external updates}, which are requests of the form
``rearrange all of memory to start at a location $k$ ahead or $k$ behind its
current start point''. Such an external update is considered an operation of
``size'' $k$, and a resizable allocator capable of handling external updates is
called \defn{relocatable}. The cost of an external update is the total size $L$
of items moved to handle the external update divided by the size $k$ of the
external update.
We now show:
\begin{restatable}{lemma}{lemmaflexhash}
\label{thm:kuszmaulflex}
If all internal updates are tiny and all external updates are large, there is a
relocatable allocator achieving worst-case expected internal update cost
$\bigO(\log \eps^{-1})$, and worst-case expected external update cost
$\bigO(1)$.
\end{restatable}

\begin{proof}
\tinyhash operates by breaking memory into \defn{slabs}, which are contiguous chunks of memory.
Let $M$ denote the largest possible size of a slab. \tinyhash satisfies $M\le \bigO(\eps^3)$.
Slabs have specific sizes and allowed locations.
In particular, slabs are of size $M/2^i$ for some $i\in \Z_{\ge 0}$.
The start locations of slabs must obey the following \defn{alignment property}:
a slab of size $L$ must be placed at a location $i\cdot L$ for some integer $i\in \Z_{\ge 0}$.
In particular, the smaller slabs nest perfectly within the larger slabs. 
We refer to intervals of the form $[M\cdot i, M\cdot (i+1)]$ for $i\in \Z_{\ge 0}$ as \defn{memory units}.
Because \tinyhash never places items spanning across memory units,
rearranging memory units doesn't break \tinyhash's correctness.

  We exploit this modular structure of \tinyhash to make a relocatable version of \tinyhash which we call \flexhash.
  \flexhash uses $\eps/2$ free space to create a \defn{buffer}. 
  \flexhash partitions the external update sizes $(\eps^4, 1]$ geometrically into $C\le \bigO(\log \eps^{-1})$
  \defn{update-types}, with the $i$-th update-type consisting of updates with size in the interval $(2^{i-1}\eps^4, 2^{i}\eps^4]$.
  \flexhash will use the buffer to ``hide'' the external updates from \tinyhash, which it will run as a subroutine. 
  \flexhash reserves the remaining $\eps/2$ free space for the normal execution of \tinyhash.

  \flexhash splits the buffer into $C$ parts, one for each update-type.
  We use the term \defn{central memory} to refer to the region of memory in which \tinyhash will operate. 
  For each $i\in [C]$ define variable $B_i$. $B_i$ stores how much of update-type $i$'s portion of the buffer has been used. \flexhash will maintain as an invariant that $B_i\in [0, 16M]$ for all $i\in [C]$ at all times.
  Furthermore, \flexhash will guarantee that the distance between 
  the start of central memory and the actual start of all of memory is at most
  $\sum_{i\in [C]} B_i$. Let $s$ denote the number of memory units that exist at
  some time. Let the $i$-th memory unit denote the memory unit that \tinyhash
  places at location $iM$. Let $\Delta$ denote the starting location of the
  memory region assigned to \flexhash. Then, there is some permutation
  $\pi:[s]\to [s]$ such that \flexhash places the $i$-th memory unit starting at
  location $\Delta + \sum_{j\in [C]}B_j+\pi_i \cdot M$. The contents of
  memory unit $i$ are identical between reality and the simulation of \tinyhash.

Any action that \tinyhash takes which happens purely
within memory units can easily be simulated by \flexhash.
To complete our description of \flexhash we must describe how \flexhash handles when \tinyhash 
creates and deletes memory units, and how \flexhash handles external updates.

\paragraph{Handling Resize operations}
\tinyhash occasionally must perform \defn{resize operations} which delete or
create memory units. When \tinyhash creates a new memory unit \flexhash also
creates a new memory unit, and places it directly after the physically final
memory unit currently in its memory.

When a resize operation destroys a memory unit for \tinyhash, it is destroying
the final memory unit for \tinyhash, and so does not create a hole in
\tinyhash's memory. However, \tinyhash's final memory unit may not correspond to
the the physically final memory unit of \flexhash. So, destroying this
memory unit might cause \flexhash to have a large hole in its memory. \flexhash
fills this hole by swapping its physically final memory unit into the location
of the deleted memory unit.

Note that swapping memory units is a quite expensive operation. Fortunately it
is at most a constant-factor more expensive than Kuszmaul's resize operations
already were. Thus, the memory unit swaps only increase \tinyhash's expected
update cost by a constant-factor.

\paragraph{Handling External Updates}
Whenever an external update occurs it ``pushes'' memory to the right or left by its size.
When an external update of update-type $i$ and size $x$ occurs we update $B_i$ in the appropriate direction based on if the update pushed memory to the right or left.

First we describe how to handle external updates of items with size at least
$M/100$. When an external update of this size occurs, if it does not break the
invariant $B_i \in [0, 16M]$ we do nothing. If it does break the invariant
\flexhash must restore the invariant. \flexhash is allowed to increase or
decrease $B_i$ by \defn{rotating} a memory unit in the appropriate direction,
i.e., taking the physically final memory unit and placing it right before the
physically first memory unit. Whenever the invariant is violated \flexhash
rotates memory units until $B_i$ is restored to being within $M$ of $8M$.
The cost of doing this on such an external update is $O(1)$.

\flexhash handles the smaller external updates by performing \defn{buffer
rebuilds} whenever there is a sufficient number of external updates of some
update-type.
Let $C'$ be the largest $i$ such that update-type $i$ consists of updates of
size at most $M/100$.
In the beginning \flexhash chooses random \defn{rebuild thresholds} $R_i, R_i'
\gets (2M, 4M)$ for each update-type $i\in [C']$.
\flexhash also initializes counters $P_i, P_i'$ to $0$; these store the total
amount that memory has been ``pushed'' in either direction by external updates
of update-type $i$.
When an external update of update-type $i$ and size $x$ pushes memory to the right or left \flexhash increases $P_i$ or $P_i'$ (respectively) by $x$.
If this causes $P_i>R_i$ or $P_i'>R_i'$ \flexhash then performs a \defn{buffer-$i$ rebuild}.
Suppose the buffer-$i$ rebuild was triggered by an external update of size $x$ that pushed memory to the right; the other case (left push) is symmetric.
To perform the buffer-$i$ rebuild \flexhash rotates memory blocks (as described in the analysis of handling large external updates) to make $B_i\in [7M, 9M]$.
Then, we set $P_i\gets P_i-R_i$ (i.e., we \emph{overflow} the unused update size to count towards the next rebuild).
Then, we randomly select $R_i \gets (2M, 4M).$
Because we always set $R_i, R_i'< 4M$ and we restore $B_i \in [7M, 9M]$ on each buffer-$i$ rebuild we clearly maintain the invariant $B_i \in [0,16M]$.
It remains to analyze the expected cost per update.
Fix some update $u$ of update-type $i\in [C']$ and size $x$.
Applying \cref{lem:hittingcontinuouscats} we find that the probability of update $u$ causing a buffer-$i$ rebuild is at most $\bigO(x/M).$
Hence, the expected external update cost is at most 
$$\frac{\bigO(M)}{x}\bigO(x/M)\le \bigO(1).$$
\end{proof}

\begin{figure}
    \centering
    \includegraphics[width=.8\linewidth]{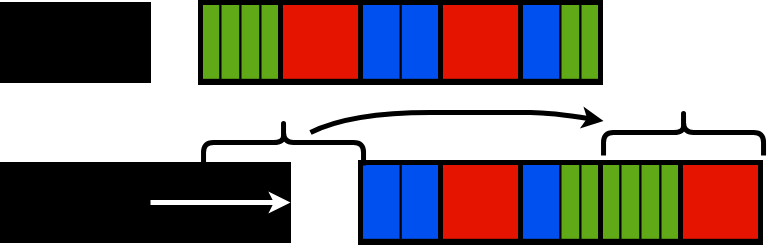}
    \caption{Because \tinyhash decomposes into interchangeable memory-units, 
    we can make \tinyhash relocatable by rotating memory-units to handle external updates.}
    \label{fig:buffers}
\end{figure}

Using the relocatable allocator \flexhash from \cref{thm:kuszmaulflex} it is easy to show:
\begin{corollary}\label{cor:onehalf}
  There is a \textbf{resizeable} allocator for arbitrary items with worst-case expected update cost $\widetilde{\bigO}(\eps^{-1/2}).$
\end{corollary}
\begin{proof}
We instantiate \geo with $\eps/2$ free space starting from $0$.
We also instantiate the relocatable \flexhash from \cref{thm:kuszmaulflex} with
$\eps/2$ free space, and we maintain the property that \flexhash starts after
\geo's memory region ends.
If there are $L_1$ total size of large items present and $L_2$ total size of
tiny items present then \geo's memory region is $[0,L_1+\eps/2]$ and \flexhash's
memory region is $[L_1+\eps/2, L_1+L_2+\eps]$.
We handle tiny items with \flexhash and large items with \geo.

 Whenever the portion of memory managed by \geo changes size by $k$ (due to an
 update of size $k$), we issue an external update of size $k$ to \flexhash in
 the appropriate direction. 
 The cost of an external update is defined precisely so that if \flexhash
 handles this external update at cost $x$ then the total size of items moved by
 \flexhash is $O(kx)$. Thus, the actual cost of this update is $O(x)$ as well.
 Hence, on any update the expected cost due to handling external updates is $O(1)$.
 The expected cost due to updates handled by \geo is at most $\widetilde{O}(\eps^{-1/2})$, and the expected cost of internal updates for \flexhash is $O(\log\eps^{-1})$.
 Thus, our allocator's expected cost is $\widetilde{O}(\eps^{-1/2})$.
\end{proof}

\section{A Lower Bound}\label{sec:short-lowerbound}
In this section we give the first non-trivial lower bound for the reallocation
problem using a surprisingly simple update sequence. 

\begin{theorem}\label{thm:ztlower}
There exist sizes $s_1,s_2 \in \Theta(\eps^{1/2})$ and an update sequence $S$
consisting solely of items of sizes $s_1,s_2$ such that any resizable allocator
(even one that knows $S$) must have amortized update cost at least
$\Omega(\log\eps^{-1})$ on $S$.
\end{theorem}
\begin{proof}
Without loss of generality assume $\eps^{-1/2} \in 4\N$, and let $n =
(\eps^{-1/2})/4$. We call the items of size $s_1$ \defn{$A$'s} and the items of
size $s_2$ \defn{$B$'s}. Set $s_1 = \eps^{1/2}+2\eps$ and $s_2 = \eps^{1/2}$.
The sequence $S$ is as follows: First, insert $n$ $A$'s. Then, for $n$
iterations, delete an $A$ and insert a $B$.

Consider an allocator operating on $S$. We will think of the allocator's
experience as follows. Every \defn{step} the allocator must rearrange memory
such that it ends with an $A$. Then, that $A$ is turned into a $B$. This is
without loss of generality because a resizable allocator cannot afford to leave a gap of
size $s_1$ in memory after an $A$ is deleted. Let the ``$i$-th item'' denote the
$i$-th item counting from the end of memory.
For $i\in [n]$ let $B_i$ denote the number of $B$'s among the final $i$ items of
memory. Define potential function (which we only measure when there are $n$
items in memory, i.e., at the start of each step): $$\Phi = \sum_{i=1}^{n}
\frac{B_i}{i}.$$ Whenever an $A$ at the end of memory is turned into a $B$, each
$B_i$ increases by $1$, so $\Phi$ increases by $\sum_{i=1}^{n} 1 / i
\ge \Omega(\log n)$. 

Now we analyze how much the allocator can change $\Phi$ by performing $x$ work.
We claim that the allocator's rearrangement can be \defn{decomposed} into
``\defn{full permutations}'', operations of the form: 
pick $i,j\in [n]$ and for each $k \in [i,j]\cap \N$ assign item $k$ a
new location. Clearly the cost of such an operation is $\Omega(j-i)$.
Intuitively this decomposition is possible because $s_1,s_2$ were constructed to
have no additive structure: for any $\lambda_1, \lambda_2 \in [0,n]\cap \Z$ not both $0$ we have
$|\lambda_1 s_1 - \lambda_2 s_2| \ge 2\eps$.
Now we show how to decompose the allocator's rearrangement into full permutations.
Fix $i,j\in [n]$ such that the allocator moves item $k$ for each $k\in [i,j]$, but does not move item $k'$ for $k'\in \set{i-1,j+1}\cap [n]$.
Let $x_1$ be the location where item $j+1$ ends (set $x_1=0$ if $j=n$) and let $x_2$
be the location where item $i-1$ starts (set $x_2=1$ if $i=1$).
Suppose that there are $a$ $A$'s and $b$ $B$'s in the memory region $[x_1, x_2]$
to start, and $a'$ $A$'s and $b'$ $B$'s in this memory region after the the
rearrangement. Note that there are no items only partially in $[x_1,x_2]$ before
or after the re-arrangement due to the assumption that the items immediately on
either side of the interval (or the endpoints of memory if no such items exist)
do not move.
As argued above, if $(a,b)\neq (a',b')$
then $|(a-a') s_1 + (b-b') s_2| \ge 2\eps$. A resizable allocator is not allowed to have more than an $\eps$ gap anywhere in memory, so this would be an invalid rearrangement. Hence we must have $(a,b)=(a',b')$. And then the allocator can simply rearrange the items within items $[i,j]$ rather than taking items from outside of $[i,j]$. Thus, we can decompose any set of rearrangements into full permutations.

Now, consider the potential change caused by a full permutation that moves $x$
items. This operation only changes the $B_i$ values for $x$ items. Thus,
because $B_i/i \le 1$ for all $i$, the operation decreases $\Phi$ by at most
$x$. This operation requires at least $x/2$ work. In summary,
the allocator requires at least $x/2$ work to decrease $\Phi$ by $x$.

We have shown a sequence of $3n$ updates such that, over the course of
the whole update sequence, $\Phi$ must increase by
$\Omega(n\log n)$. Since the potential starts at $0$ and is
always at most $n$, the allocator must have amortized cost at least
 $$\frac{1}{2}\frac{\Omega(n\log n) - n}{3 n} \geq \Omega(\log n)\ge \Omega(\log\eps^{-1}).$$
\end{proof}

%%%%%%%%%%%%%%%%%%%%%%%%%%%%%%%RANDOM SECTION%%%%%%%%%%%%%%%%%%%%%
\section{An Allocator for Items with Random Sizes in $[\delta, 2 \delta]$}\label{sec:random}
In this section we consider allocators for \defn{random items}, i.e., items with uniformly random sizes in some range $[\delta, 2\delta]$.
In this setting we are able to create an allocator with
substantially better performance than the allocators of \cref{sec:geo}.

Fix $\delta = \poly(\eps)$. 
% It is helpful for 
% For intuition's sake one can  consider the example $\delta=\eps$ in mind, but we can also handle larger  values of $\delta$, e.g., $\delta=\sqrt{\eps}$, with some extra ideas.
A $\delta$-\defn{random-item sequence} is the following sequence of updates:
The first $\floor{\delta^{-1}/4}$ updates are inserts of items with sizes chosen
randomly from $[\delta, 2\delta]$. Then, the sequence alternates between a deletion of a random item and an insertion of an item with size chosen randomly from $[\delta, 2\delta]$.
Note that there will always be (within $1$ of) $\floor{\delta^{-1}/4}$ items present. 
Our main result of this section is:
\begin{theorem}\label{thm:summer}
There is a randomized resizable allocator that handles $\delta$-random-item sequences with
worst-case expected update cost  $\bigO(\log \eps^{-1})$. Furthermore, the set
of items that our allocator moves to handle an update can be computed in
expected time $\bigO(\eps^{-1/2})$.
\end{theorem}

Note that in this stochastic setting where the total size of items present is variable 
the resizable guarantee of our allocator is the most natural property to hope for.
To prove \cref{thm:summer}, the following property of $\delta$-random-item sequences is quite useful:
After $d\ge \floor{\delta^{-1}/4}$ updates the distribution of
items sizes present is the same distribution as obtained by sampling
$\floor{\delta^{-1}/4}$ (or $\floor{\delta^{-1}/4}+1$ depending on the parity of $d$)
values independently from $[\delta, 2\delta]$.
% This property is quite powerful. However, great care is still required:
% although the distribution at any point in time is the 
% same as this uniform distribution there is
% very large correlation between the items present at any two nearby times. 

Our allocator for random items is based on the observation that
random independent values can make many subset sums. 
The subset sums of random sets have been studied before (see, e.g., \cite{Lu98}).
However, to the best of our knowledge previous work has only given an
asymptotic version of the result we need, namely \cref{thm:randomthm}.
Our self-contained analysis explicitly determines the constant-factor for how
large a random set has to be in order to contain a subset of a desired sum with constant probability.
This is important for our application because the constant-factor
appears as an exponent in the running time of our algorithm.

In what follows our analysis is asymptotic in a parameter $n\in\N$ (rather than
in $\eps^{-1}$ like in all other places in the paper).
First we need a standard fact about sums of random variables. 
We show in \cref{sec:appendixobvious} how to derive this fact from a theorem in \cite{Uspensky1937}.
\begin{fact}\label{fact:ERFC}
Fix constants $a,b>0$.
  Let $x_1,\ldots, x_n \gets [0,1]$ be chosen uniformly
  randomly and independently. Then
  $$\Pr\left[\sum_{i=1}^{n} x_i \in [n/2-a,n/2+b] \right] =
  \Theta(1/\sqrt{n}).$$
\end{fact}

We will also need the following asymptotic expression for
binomial coefficients (see, e.g., \cite{zhao2023graph}):
\begin{fact}\label{fact:entropy}
  Define the binary entropy function $H$ as\\
  $H(x) = -x\log x -(1-x)\log (1-x)$.
   For any constant $\alpha\in (0,1),$ 
  $$\binom{n}{\ceil{\alpha n}} =  \Theta\left(2^{nH(\alpha)}/\sqrt{n}\right).$$
\end{fact}
We establish the following theorem:
\begin{theorem}\label{thm:randomthm}
  Let $m= 2\lceil (\log n)/2\rceil$.
  Fix arbitrary $y\in (3/4)m + [-1,1]$. 
  Let $x_1,\ldots, x_m  \gets
  [1,2]$ be uniformly random and independent values.
  Then, with probability $\Omega(1)$ there exists an $(m/2)$-element subset of
  $x_1,\ldots, x_m$ with sum in $[y-\frac{\log n}{n},y]$.
\end{theorem}
\begin{proof}
Let $\I_y = [y-\frac{\log n}{n}, y]$.
Let random variable $S$ denote the number of $(m/2)$-element
subsets of $x_1,\ldots, x_m$ with sum in $\I_y$. 
\begin{lemma}\label{lem:ESK} 
$\E[S]\ge \Omega(1).$
\end{lemma}
\begin{proof}
Let $z_1, z_2, \ldots, z_{m/2}$ be sampled uniformly from $[1,2]$.
Define random variable $Z = \sum_{i=1}^{m/2-1}z_i$.
Let $Z^*$ denote the event $Z\in [y-2, y-1-\frac{\log n}{n}]$.
Then,
$$\Pr\left[Z+z_{m/2} \in \I_y \right] \ge \Pr[Z^*] \cdot \Pr[Z+z_{m/2}\in\I_y \mid Z^*].$$
Bounding the probability in this manner is productive because $Z^*$ is very
likely, and conditional on $Z^*$ the event $Z+z_{m/2}\in \I_y$ is easy to analyze.
In particular, 
$$\E[Z] = (m/2-1)\cdot (3/2) \in [y-3, y+3].$$
Thus \cref{fact:ERFC} implies that $\Pr[Z^*] = \Theta(1/\sqrt{m})$.
If $Z^*$ occurs, making the $(m/2-1)$-th partial sum very close to the desired
value, then with probability $\frac{\log n}{n}$ the value of
$z_{m/2}$ makes $Z+z_{m/2}\in\I_y.$
So we have found
\begin{equation}
\Pr[Z+z_{m/2}\in \I_y] \ge \Omega\left(\frac{\log n}{n \sqrt{m}}\right)\ge \Omega\left(\frac{\sqrt{\log n}}{n}\right). \label{eq:ZzIy}
\end{equation}
Now we use \cref{eq:ZzIy} to show $\E[S]$ is large.
  Using linearity of expectation over all $\binom{m}{m/2}$ possible
   $(m/2)$-element subsets of the $x_i$'s we conclude:
   $$\E[S] \ge \Omega\left(\frac{\sqrt{\log n}}{n}\right) \cdot\binom{m}{m/2} \ge \Omega\left(\frac{\sqrt{\log n}}{n}  \cdot \frac{2^{\log n}}{\sqrt{\log n}}\right) \ge \Omega(1).$$
\end{proof}

Let $A_1$ denote a uniformly random value from $[1,2]$ and for
each $i\in \N$ let $A_{i+1}$ denote $A_i$ plus another random
independent value drawn from $[1,2]$.
\begin{lemma}\label{lem:independence_is_real}
  For any constant $\lambda \in (0,1)$, any $i\in \N$ with $i\le \lambda m/2$ and
  any $a\in \R$ we have
  $$\Pr[A_{m/2}\in \I_y \mid A_i=a] \le
  O\left(\frac{\sqrt{\log n}}{n}\right) .$$
\end{lemma}
\begin{proof}
  By \cref{fact:ERFC} we have that for any value of $A_i$ there
  is at most a $\bigO(1/\sqrt{m/2-i})\le O(1/\sqrt{m})$ chance that $A_{m/2-1}$ sums
  to within $2$ of $y$. Conditional on $A_{m/2-1}$ being this close
  to $y$ there is at most a $\frac{\log n}{n}$ chance that the value added to $A_{m/2-1}$ to make $A_{m/2}$ makes the sum $A_{m/2}$ precisely lie in the interval $\I_y$.
  Multiplying these probabilities yields the desired bound.
\end{proof}

We now proceed with the proof of the theorem.
We will use the second moment method (\cite{alon2016probabilistic}) to show that $\Pr[S>0] \ge \Omega(1).$

Let $\mathcal{X}$ denote the set of all size-$m/2$ subsets of $[m]$.
  For $A\in \mathcal{X}$ let indicator variable $S_A\in
  \set{0,1}$ indicate the event that $\sum_{i\in A} x_i \in \I_y.$ 
  Of course 
  $S = \sum_{A\in\mathcal{X}} S_A.$
  Let $\lambda = 4/5$.
  We decompose $\E[S^2]$ as:
  \begin{multline}\label{eq:esk2}
  \E[S^{2}] = \sum_{\substack{A,B\in \mathcal{X}^2\\ A=B}}\Pr[S_A\land S_B] +
  \sum_{\substack{A,B\in \mathcal{X}^2\\ |A\cap B|<\lambda m/2}}
  \Pr[S_A\land S_B]\\
  + \sum_{\substack{A,B\in \mathcal{X}^2\\ \lambda m/2\le |A\cap B|<m/2}} \Pr[S_A\land S_B].
  \end{multline}
  Let $T_1,T_2,T_3$ denote the three terms in \cref{eq:esk2} in the order they appear. 
  $T_1$ is simply $\E[S]$. 
  Recall that \cref{lem:ESK} says $\E[S]\ge \Omega(1)$. 
Thus,  
\begin{equation}\label{eq:term1}
T_1=\E[S]\le \bigO(\E[S]^{2}).
\end{equation}

We can bound the probability in the sum defining $T_2$ using \cref{lem:independence_is_real}.
In particular, observe that $A\cap B$ is a sufficiently small set, so if we
condition on $\sum_{i\in A\cap B} x_i$ the conditional probabilities of $S_A, S_B$ are at most $O((\sqrt{\log n})/n)$. The number of terms in the sum defining $T_2$ is trivially at most $|\mathcal{X}|^2$. Thus, we have
\begin{equation}\label{eq:term2}
  T_2 \le O\left(\frac{\log n}{n^2}\right) \cdot \binom{m}{m/2}^2 \le \bigO(1).
\end{equation}

The probabilities of $S_A,S_B$ in the sum defining $T_3$ might be highly
correlated so we cannot use the strong bound that we used when bounding $T_2$.
Fortunately, for $A\neq B$ in order for both events $S_A,S_B$ to occur we need
two distinct random values to land in specific intervals of size
$\frac{\log n}{n}$. Specifically if $A\neq B$ then we can find $i_B\in
B\setminus A$ and $i_A\in B\setminus A$. Then, after conditioning on the value
of $x_i$ for each $i\in [m]\setminus\set{a,b}$ the probability that
$S_A$ and $S_B$ both occur is at most $\frac{\log^2 n}{n^2}.$
Fortunately the number of terms in the sum defining $T_3$ 
is not too large: it is at most
$$\binom{m}{\ceil{m\lambda/2}}\binom{m}{m/2-\ceil{m\lambda/2}}^2,$$
because we can first chose $A\cap B$ and then chose $A\setminus B, B\setminus A$.
Thus,
\[
  T_3\le \frac{\log^2 n}{n^2}\binom{m}{\ceil{\lambda m/2}} \binom{m}{m/2-\ceil{m\lambda/2}}^2.
\]
Now we show $T_3<o(1)$. Using \cref{fact:entropy} we have
\begin{equation*}\label{eq:gross11}
  T_3 \le O\left( \frac{\log^2 n}{n^2} 2^{mH(\lambda/2)}2^{2mH((1-\lambda)/2)}
\frac{1}{\left(\sqrt{\log n}\right)^{3} }\right).
\end{equation*}
Thus
\begin{equation*}
  \log T_3\le
  \bigO(1)+\log\log n + \left(H(\lambda/2)+2H\left(\frac{1-\lambda}{2}\right)-2\right)\cdot \log n.\label{eq:letsmakeitnegative}
\end{equation*}
Evaluating the expression with $\lambda = 4/5$ we find  $\log T_3< -\Omega(\log n)$. 
Thus $T_3<o(1)$ as desired.

Now we combine our bounds on $T_1,T_2,T_3$ to obtain, via the second moment method (see chapter 4 of \cite{alon2016probabilistic}), the bound
\[
\Pr[S>0] \ge \frac{\E[S]^2}{\E[S^2]} = \frac{\E[S]^2}{T_1+T_2+T_3} \ge \Omega(1).
\]

\end{proof}

We are now equipped to prove \cref{thm:summer}.
\begin{proof}[Proof of \cref{thm:summer}]
We call our allocator \trash.
We start by giving a construction that works if $\delta\le \eps/4$. At the end
of the proof we show how to modify this construction to work in the case
$\delta >\eps/4$ as well. \trash reserves $\eps/2$ free space for use as a
\defn{buffer}, which will separate the \defn{main-body} of memory and the
\defn{trash can}: a suffix of the used portion of memory. The trash can, buffer and main-body all
start empty. It is important that the buffer is at least the size of the
largest item, i.e., $\eps/2\ge 2\delta$. If this is not the case we will need
a more involved construction for the buffer; we discuss this at the end of the
proof. \trash reserves the remaining $\eps/2$ free space to enable \trash to create waste by
introducing up to $\delta^{-1}/(2\log\eps^{-1})$ gaps of size up to $g=
\eps \delta \log \eps^{-1}$ in memory.

\trash operates somewhat similarly to the \geo allocator of \cref{sec:geo} in
that \trash handles deletes by performing \defn{swaps} that introduce small
amounts of waste in memory, and periodically \defn{rebuilds} memory to eliminate this waste. 
The main difference between \trash and \geo is that \trash swaps
\emph{sets} of items rather than single items. This gives \trash much greater
flexibility, resulting in its substantially lower cost.

\trash groups the items in the main-body into \defn{blocks} of $m=2\ceil{ (\log
\eps^{-1})/2 }$ items; the items in the trash can are not part of blocks. 
Blocks will be the basic units that facilitate \trash's swap operations.
Blocks are marked as either \defn{valid} or \defn{invalid}.

\paragraph{Handling Deletes}
Suppose an item $I$ is deleted. 
\trash forms a set $Y$ containing $I$ and roughly $m/2-1$ other nearby items,
with total size $y \in \frac{3}{4}m\delta + [-\delta, \delta]$.
In particular, if $I$ is in the main-body \trash arbitrarily adds items
contiguous with $I$ from the same block to $Y$ until the total size of $Y$ lies
in $\frac{3}{4}m\delta+[-\delta, \delta]$.
If $I$ is in the trash can \trash simply adds arbitrary trash can items contiguous with
$I$ to $Y$ until the total size is appropriate. 
Constructing $Y$ is possible because items have size at most $2\delta$.

\trash then attempts to find a block $B$ near the end of the main-body with a
subset of elements $S$ whose sum $z$ is in the interval $[y-g, y]$.
We say that such a block is \defn{compatible} with $Y$.
To find a compatible block \trash\xspace \defn{checks} whether the final valid block in the
main-body is compatible with $Y$. If it is not \trash invalidates this block and
keeps trying valid blocks. 
If the number of valid blocks ever becomes too small \trash will abandon its
search for a compatible block and instead handle the delete via a \emph{rebuild
operation}, which will be described later. 
So we may assume \trash finds a valid block $B$ with corresponding subset $S$ of sum
$z\in [y-g, y]$.

\trash now swaps $S,Y$. 
To swap $S,Y$, \trash first takes items $S$ and
arranges them contiguously in the region of memory where items $Y$ used to be,
leaving a gap of size at most $g$. \trash then takes items $Y\setminus \set{I}$
and items $B\setminus S$ and arranges them contiguously in the region of memory
that was occupied by block $B$. We remark that if $I$ is part of block $B$ 
the above steps do nothing. \trash then removes $I$ from memory.
Once a block of items has been used for a swap \trash marks the block as
invalidated. In particular, both the block $B$ used to repair the delete
and the block where the delete occurred (if $I$ was in the main body) are invalidated.

To finish the swap \trash\xspace \defn{pushes} some blocks into the trash can.
Recall that the trash can is a suffix of the utilized portion of memory,
separated from the main-body by a small buffer. Once an item $I_0$'s block has
been invalidated \trash may place $I_0$ in the trash can.  However, invalidated
blocks need not be immediately placed in the trash can. 
When a swap happens, taking items $S$ from block $B$ to repair a delete, \trash
takes block $B$ and all blocks to its right in the main-body and moves them to
be contiguous with the start of the trash can, and compacts them against the
start of the trash can. At this point \trash no-longer considers items from
these pushed blocks to be part of any blocks. 

When \trash pushes blocks into the trash can it will potentially increase the
size of the buffer (i.e., the distance between the trash can and the main body)
due to the empty space created by removing item $I$ from memory.
If the buffer size now exceeds $\eps/2$ \trash takes items from the end of the
trash can and rotates them to be flush with the beginning until the buffer size
is again at most $\eps/2$. 

\paragraph{Handling Inserts}
\trash handles inserts by placing the inserted item after the final item currently in memory and adding the inserted item to the trash can.

\paragraph{Performing Rebuilds}
In addition to responding to deletes and inserts as described above \trash
occasionally must perform expensive \defn{rebuild operations} that ensure
necessary guarantees on the layout of items in memory.

In the beginning \trash uniformly randomly samples a \defn{rebuild threshold} ${r\gets
(\delta^{-1}/(8m), \delta^{-1}/(6m))\cap \N}$. 
This counts as a ``free rebuild'' (as a sentinel value).
If an update would cause the number of valid blocks to drop below $r$,
instead of handling it normally \trash randomly permutes all items, places them
contiguously into memory to eliminate all waste, and then logically partitions
the items into blocks of $m$ contiguous items, starting from the right of
memory. \trash then resamples $r$.

\begin{figure}
    \centering
    \includegraphics[width=1\linewidth]{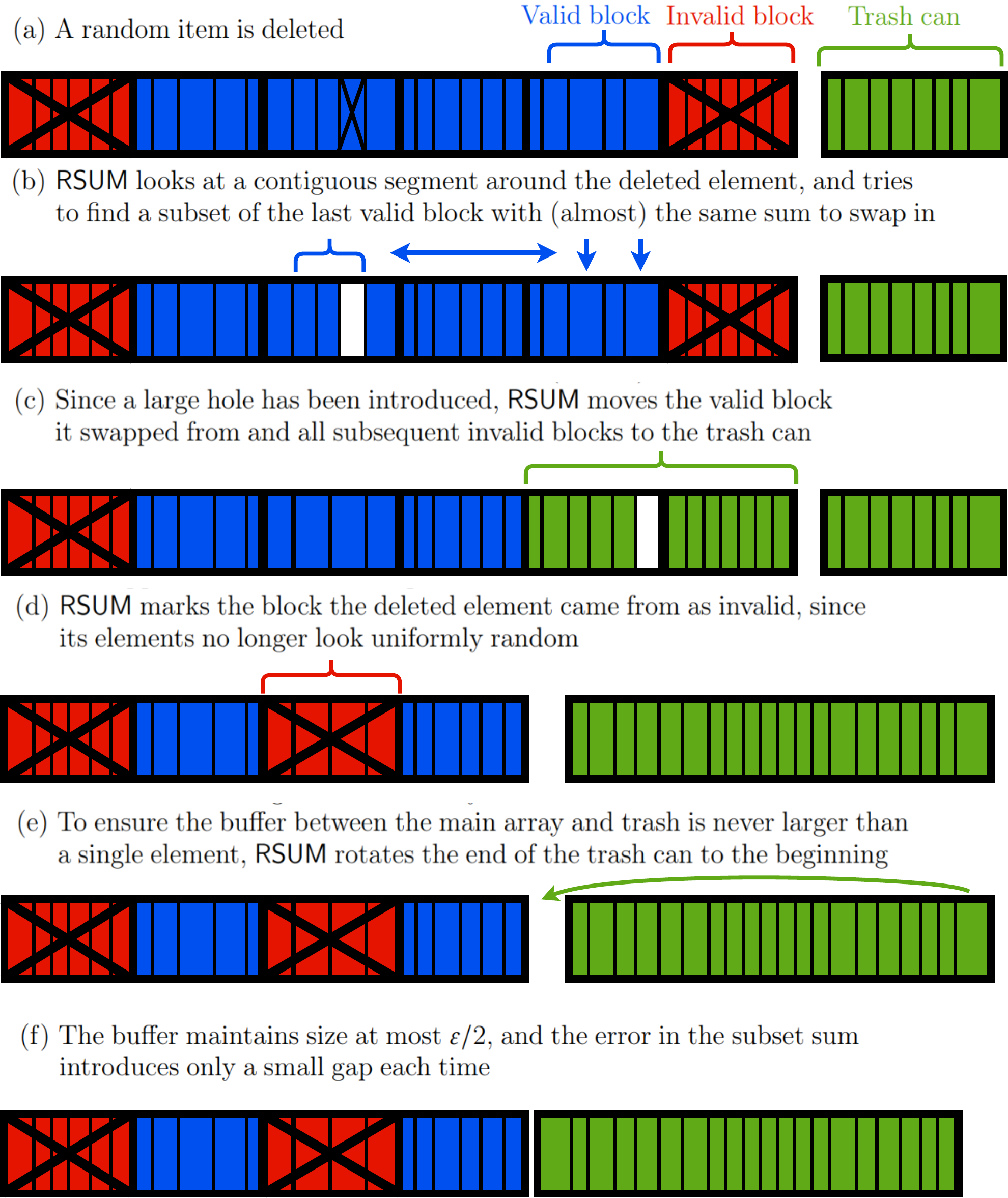}
    \caption{Operation of \trash. Since valid blocks have never been modified, their elements will look uniformly random -- by our analysis, we know they're large enough that one of the last couple valid blocks has a subset sum with size very close to the neighbourhood of the deleted element. So, \trash swaps that subset in and then pushes a suffix of memory into the trash can to eliminate the hole.}
\end{figure}

  \begin{algorithm}
        \caption{Random-Item Sequence: \trash \\ Assume $\delta < \eps/4$}\label{alg:random}
        \begin{algorithmic}[1]
          \State Choose random rebuild threshold $r \in \left(\frac{\delta^{-1}}{8m}, \frac{\delta^{-1}}{6m}\right)\cap \N$, initialize the trash can, buffer and main-body to be empty and consider this a ``free rebuild''.
          \If{\emph{at any time} there are fewer than $r$ remaining valid blocks}
          \State Perform a rebuild:
          \State Stop the current operation.
          \State Compact all the items, eliminating all gaps. 
          \State Randomly permute the items.
          \State Logically partition the items into \defn{blocks}
          of $m$ contiguous items, starting from the end of memory.
          \State Mark all blocks as \defn{valid}.
          \State Set the trash can to be empty.
          \State Resample $r \in \left(\frac{\delta^{-1}}{8m}, \frac{\delta^{-1}}{6m}\right)\cap \N$
          \EndIf

          \If{an item $I$ is inserted}
          \State Place $I$ immediately after the currently final item and add $I$ to the trash can.
          \ElsIf{an item $I$ is deleted}

          \State Let $Y$ be a set of items contiguous with $I$, including $I$, whose total size $y$ satisfies $y\in \frac{3}{4}m\delta + [-\delta,\delta]$; this is possible because the maximum item size is $2\delta$.
            If $I$ is not in the
          trash can choose these items to be in the same block as $I$.

          \While{there is no subset of the final valid block with
          total size in $[y-g,y]$}
          \State Invalidate the final valid block.
          \EndWhile
          \State Let $B$ be the final valid block and let $S$ be
          a subset of $B$ with total size in $[y-g,y]$.
          \State Take items $S$ and arrange them contiguously in the region of
          memory where items $Y$ used to be, leaving a gap of size at most $g$.
          \State Take items $(Y\setminus \set{I}) \cup (B\setminus S)$ and
          arrange them contiguously in the region of memory that was occupied
          by block $B$
          \State Invalidate all blocks involved in the swap.
          \State Remove item $I$ from memory, introducing a gap
          in block $X$.
          \State Take the block $X$ and all blocks to its
          right not yet in the trash can and move all of these
          items into the trash can, compressing them to eliminate
          gaps between them.
          \If{the buffer has now grown too large}
          \State Take an item from the end of the trash can and
          swap it to the beginning of the trash can.
          \EndIf
          
          \EndIf
        \end{algorithmic}
    \end{algorithm}

We give pseudocode for \trash in \cref{alg:random}.
Now we verify that \trash is well-defined and analyze \trash's performance.
\begin{lemma}
   \trash places items in valid locations.
\end{lemma}
\begin{proof}
Note that the items present are always kept contiguous except for small gaps
introduced by swaps and the buffer between the main-body and the trash can. Each
swap creates wasted space at most $g$ and invalidates at least $1$ block. 
The total number of blocks is $\floor{\ceil{\delta^{-1}/4}/m}$. 
\trash certainly rebuilds before all blocks are invalidated. Hence, the wasted space never exceeds
$\floor{\ceil{\delta^{-1}/4}/m} g \le \eps/2$.
\trash regulates the size of the buffer to be at most $\eps/2$, so this ensures
that if there is $L$ total size of items present at some point in time then the
items fit in the space $[0,L+\eps]$. 
\end{proof}

\begin{lemma}\label{lem:rsumanalysis}
\trash's worst-case expected update cost is $\bigO(\log \eps^{-1})$.
The set of items to move at each update
by \trash can be computed in expected time $\bigO(\eps^{-1/2})$.
\end{lemma}
\begin{proof}
\trash clearly has cost $O(1)$ per insert.
Before analyzing the expected cost of deletes, we analyze the rate at which
blocks are invalidated: this will dictate the cost of rebuilds.

We will show using \cref{thm:randomthm} that in expectation only $O(1)$ valid
blocks must be checked before finding a compatible valid block to handle each delete.
Fix some delete.
Let $y\in \frac{3}{4}m\delta + [-\delta, \delta]$ be the size of the set of
items $Y$ contiguous with the deleted item which we aim to swap. 
Suppose we are given a set $X$ of $m$ items with sizes 
chosen uniformly randomly and independently from $[\delta, 2\delta]$.
We claim that with constant probability there is a subset $X'\subset X$ such that 
$\sum_{x\in X'} x \in [y-g, y]$. 
This follows immediately from \cref{thm:randomthm}, with all sizes scaled down by a factor of $\delta$.

Intuitively this means that the expected number of valid blocks \trash looks at
on each delete should be $\bigO(1)$. Now we formalize this intuition.
Define a \defn{phase} to be the set of updates between rebuild steps.
Note that the set of items present is highly correlated between phases, so great
care is needed. 
However, we will argue that \trash's periodic rebuild operations, where \trash
randomly permutes all present items, guarantee the following property:
Let $C_i$ denote the event that the $i$-th check of a valid block's compatibility during a fixed phase succeeds. Then for all distinct $i,j$ the events $C_i,C_j$ are independent and occur each with probability $\Omega(1)$.
We call this property the ``\defn{purity of valid blocks}''.

We now argue why the purity of valid blocks property holds. If a block is valid,
it means that \trash has not touched or even looked at the items in the block
during the phase so far. Since the set of items sizes present at the start of
the phase is equivalently distributed to randomly sampled items, the sizes of
the items in each valid block is equivalently distributed to randomly sampled
items, as their randomness has not been spoiled. Thus, the events $C_i$ are
indeed independent random variables, and occur with probability $\Omega(1)$ by
the argument above (i.e., applying \cref{thm:randomthm}).

Thus, the expected number of blocks that \trash invalidates on each delete is
the expectation of a geometric random variable with probability $\Omega(1)$ of
occurring and hence is $\bigO(1)$. In particular this implies that the expected
number of steps before there are fewer than $r$ valid blocks is $\Omega(\delta^{-1}/m)$.  

Now we analyze the expected cost of update $u$.
There are four costs that we must analyze: the costs due to (1) rebuilding, 
(2) swapping items to handle deletes, (3) pushing items into the trash
can, and (4) rotating items to make the buffer sufficiently small.
Intuitively, because each phase has expected length $\Omega(\delta^{-1}/m)$ 
and because the rebuild threshold $r$ is random, the
expected cost of rebuilding per update is $O(\log\eps^{-1})$; we formalize this \cref{lem:scarycats} after discussing the other costs.

The swap operation has cost $O(m)\le O(\log\eps^{-1})$ because there are $O(m)$
items amongst the two blocks involved in the swap.
Repairing the buffer has cost $O(1)$: it requires moving at most $O(1)$ items.
Now we analyze the cost of pushing items into the trash can.
Using the purity of valid blocks property we have that every delete
decreases the number of valid blocks by at most $\bigO(1)$ in expectation. 
Since \trash always rebuilds before the number of valid blocks drops below
$\delta^{-1}/(8m)$ at most $1/2$ of the blocks are invalid at any
point. Since the delete locations are uniformly random, the subset of blocks that
are invalid is uniformly distributed in the main-body, conditional on its size.
Thus, in expectation the number of blocks that \trash must push to the trash can
on update $u$ is at most twice the number of blocks it invalidates.
As \trash invalidates $O(1)$ expected blocks, it 
only pushes $O(1)$ expected blocks to the trash can in total, for which it incurs cost $O(\log\eps^{-1})$.

Now we formally analyze the expected cost due to rebuilding.
\begin{lemma}\label{lem:scarycats}
    The expected cost of rebuilding on update $u$ in \trash is $O(\log\eps^{-1})$.
\end{lemma}
\begin{proof}
Fix some update $u$. Let $L = \delta^{-1}/(8\log\eps^{-1})$ be the maximum
number of blocks that can be invalidated during a phase. Note that the minimum number of blocks that must be invalidated in each phase is $\delta^{-1}/(12\log\eps^{-1}) = 2L/3$.
Recall that by the purity of valid blocks property the random variables
$C_i$, indicating whether the $i$-th check for compatibility succeeds within
some fixed phase are independent and each occur with probability $p\ge \Omega(1)$.
Thus, by a Chernoff Bound, the \defn{length} of each phase, i.e., the number of
updates that are handled during that phase, is at least $Lp/50$ with probability
$1-e^{-\Omega(-L)}$. 
Thus, with exponentially good probability there are at most $\ceil{50/p}$
rebuilds during the interval $[u-L, u]$.
For each $i\in [\ceil{50/p}]$ the chance that $u$ is responsible for the $i$th rebuild in $[u-L, u]$ is at most $O(1/L)$. Applying a union bound we see
that update $u$ is responsible for a rebuild with probability at most $O(1/L)$. Note that it is impossible for $u$ to be responsible for multiple rebuilds. Thus, the expected cost of performing a rebuild on update $u$ is at most
$$O(\delta^{-1}/L)\le O(\log\eps^{-1}).$$
\end{proof}

Now we analyze the expected running time required to compute \trash's strategy.
The running time is dominated by the expected $O(1)$ times that \trash must
check if a valid block is compatible to handle the delete. Each such check can
be performed by computing all subset sums of the $m$ item sizes in the valid
block that it is checking. This requires $\bigO(\eps^{-1/2})$ time by using the
meet-in-the-middle algorithm for finding subset sums.
\end{proof}

In the above analysis we have assumed $\delta\le \eps/4$ for simplicity of
exposition. The only place we used this assumption is in constructing the buffer
that separates the trash can from the main-body: a simple buffer requires an
items-worth of slack. 
We now show how to handle the
regime $\delta > \eps/4$ as well. 

\begin{lemma}
\trash can be modified to work for $\delta > \eps/4$.
\end{lemma}
\begin{proof}
We now describe a more complicated buffer management strategy that allows \trash to handle $\delta>\eps/4$.
Fix some delete.
After pushing items into the trash can on this delete \trash\xspace \defn{stashes} the
final valid block from the main-body: that is, \trash temporarily considers this
block to not be contained in memory. 
Then \trash rotates items from the back of the trash can to the front until the
distance between the main-body and the start of the trash can is some value
$y \in (3/4)m\delta +[-\delta,\delta]$. 
Then, \trash attempts to find a subset $S$ of the stashed block summing to a
value in the range $[y-\eps/2, y]$. 
\trash will succeed with constant probability due to \cref{thm:randomthm}, which
applies because $\eps/2 > g$ due to $\delta=\poly(\eps)$.
If \trash fails, it invalidates the stashed valid block, pushes it and all
blocks to its right into the trash can, and redoes the the stashing and cycling
steps from above using the next valid-block. 
By the same analysis as in \cref{lem:rsumanalysis} in expectation it takes at
most $O(1)$ tries before \trash successfully finds a valid block with a subset
$S$ summing to a value in the range $[y-\eps/2, y]$.
Suppose \trash finds a block with items $B$ that has a subset $S$ with the desired sum.
\trash then places all items $B$ in the trash can.
However, \trash places items $S$ at the front of 
the trash can and items $B\setminus S$ at the end of the trash can. 
Then, the gap between the main-body and the trash can is of size at most
$\eps/2$, as desired.

Clearly this more complex buffer management scheme increases the cost of \trash's
updates and the time required to compute \trash's strategy by at most constant factors.

\end{proof}

\end{proof}

\section{Conclusion}
Our main contribution in this paper is an allocator for the memory reallocation problem achieving expected update cost $\widetilde{\bigO}(\eps^{-1/2})$. 
However, there are several indications that it should be possible to construct an allocator with much lower expected update cost.

Kuszmaul has already established that if all items are smaller than $\eps^4$
then there is an allocator with expected update cost $O(\log \eps^{-1})$.
Using similar techniques to the covering sets introduced in this paper one can see that there are efficient allocators for sets of items with few distinct sizes and where and all sizes are fairly similar. Combined with the standard technique of discretizing item sizes this approach becomes even more powerful.
Thus, ``structured'' sets of items can be handled efficiently.
On the other hand, we gave an allocator that achieves update cost $\bigO(\log
\eps^{-1})$ for large stochastic items.
Thus, it seems plausible that there is a ``structure versus randomness'' dichotomy that can be exploited to achieve better allocators for arbitrary items.
We leave constructing an allocator with expected update cost $o(\eps^{-1/2})$, or strengthening our lower bound, as open problems.

% results that aren't in this paper: (see extra_small_items.tex, extra_zt_stuff.tex)
% someday I will polish these and add them to an appendix!
% 1. "FISH" algorithm for t-color zero-tolerance
% 2. $\O(\log \eps^{-1})$ cost if all item sizes are smaller than $\eps^{2.01}$.
% 3. Our allocator can be de-randomized (i.e., remove the word expected from it)
% and the open problems listed are indeed intruiging, definitely worth coming back to at some point.

\bibliographystyle{ACM-Reference-Format}
\bibliography{refs}

\appendix
\section{Omitted Lemmas}\label{sec:appendixobvious}
In this section we prove several lemmas omitted from the main paper.

\begin{reptheorem}{fact:ERFC}
Fix constants $a,b>0$.
  Let $x_1,\ldots, x_n \gets [0,1]$ be chosen uniformly
  randomly and independently. Then
  $$\Pr\left[\sum_{i=1}^{n} x_i \in [n/2-a,n/2+b] \right] =
  \Theta(1/\sqrt{n}).$$
\end{reptheorem}
\begin{proof}
According to \cite{Uspensky1937} there exists an absolute constant $C>0$ such
that the following holds.
Suppose $y_1,y_2,\ldots, y_n \gets [-1/2,1/2]$ are chosen uniformly randomly and
independently. Then for any $t>0$ we have:
\begin{equation}\label{eq:fancycontinuousequation}
\left|  \Pr\left[\left|\sum_{i=1}^n y_i \right| \le t \sqrt{n}\right]  - C \int_{-t}^t e^{-u^2/2} du \right| \le \bigO(1/n).
\end{equation}
Let $c\in \set{a,b}$ and take $t=c/\sqrt{n}$.
Then
$$\int_{-t}^{t}e^{-u^2/2} du = \Theta(1/\sqrt{n}).$$
Using this in \cref{eq:fancycontinuousequation} we find 
$$\Pr\left[\left|\sum_{i=1}^n y_i \right| \le c\right] = \Theta(1/\sqrt{n}).$$
By symmetry we then have
$$\Pr\left[\sum_{i=1}^n y_i  \in [0,c]\right] = \Pr\left[\sum_{i=1}^n y_i  \in [-c,0]\right] = \Theta(1/\sqrt{n}).$$
Summing the probability of the sum landing in either of $[0,b], [-a,0]$ and
translating the $y_i$'s by $+1/2$ and gives the desired bound. 
\end{proof}

\continuousCats*
\begin{proof}
If $b-a\ge W/4$ the statement is vacuously true. 
So, we may assume $b-a< W/4$.

For each $z\in \R$ let $\H(z)$ denote the event that there exists $j$ with
$\sum_{i\le j} x_i = z.$ Let $\H(a,b)$ denote the event that there exists $j$
with $\sum_{i\le j} x_i \in [a,b]$. Observe that $\H(a,b) = \int_a^b \H(z) dz.$
If $b\le W$ then $\Pr[\H(a,b)] \le 2(b-a)/W$. Suppose $b>W$.

In order for $\H(a,b)$ to happen the following must occur: there must be $j\in
\N$ and $z\in (a-W, b-W/2)$ such that $\sum_{i\le j}x_i = z$, and then $x_{j+1}$
must satisfy $z+x_{j+1}\in [a,b]$. Observe that events $\H(z), \H(z')$ are
disjoint if $|z-z'|\le W/2.$
Thus we have:
\[\Pr[\H(a,b)]\le \int_{a-W}^{b-W/2} \frac{2(b-a)}{W} \H(z) dz \le \frac{4(b-a)}{W}.\]
\end{proof}

\discreteCats*
\begin{proof}
    For any $z\in \Z$ let $\H(z)$ denote the event that there exists $j$ with
    $\sum_{i\le j}x_i = z$.
    If $\H(y)$ then we must either have $y\le \ceil{N/3}$ or else $\H(y-i)$ is
    true for some $i\in [\ceil{N/4}, \ceil{N/3}]$.
    In the case that $y\le \ceil{N/3}$ we clearly have $\Pr[\H(y)] \le 100/N$.
    Now, suppose $y>\ceil{N/3}$.
    Observe that $\setof{\H(y-i)}{i\in [\ceil{N/4}, \ceil{N/3}]}$ are disjoint events. Thus,
    $$\sum_{i=\ceil{N/4}}^{\ceil{N/3}} \Pr[\H(y-i)] \le 1.$$
    Thus we have
    \[ \Pr[\H(y)] \le \sum_{i=\ceil{N/4}}^{\ceil{N/3}} \frac{\Pr[\H(y-i)]}{\ceil{N/3}-\ceil{N/4}+1} \le 100/N. \]
\end{proof}

\end{document}